\documentclass{article}

\usepackage{arxiv}

\usepackage[utf8]{inputenc} % allow utf-8 input
\usepackage[T1]{fontenc}    % use 8-bit T1 fonts
\usepackage{url}            % simple URL typesetting
\usepackage{booktabs}       % professional-quality tables
\usepackage{amsfonts}       % blackboard math symbols
\usepackage{nicefrac}       % compact symbols for 1/2, etc.
\usepackage{microtype}      % microtypography
\usepackage{lipsum}		% Can be removed after putting your text content
\usepackage{graphicx}
\usepackage{natbib}
\usepackage{doi}
\usepackage{algorithm}

\usepackage{times}
\usepackage{soul}
\usepackage{hyperref}
\usepackage[small]{caption}
\usepackage{amsmath}
\usepackage{amsthm}
\usepackage[switch]{lineno}
\usepackage{endnotes}

\author{ \href{https://orcid.org/0000-0001-9729-2959}{\includegraphics[scale=0.06]{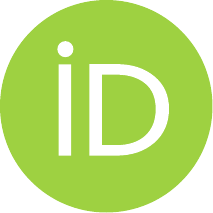}\hspace{1mm}Hugo Gilbert}\\
Université Paris-Dauphine, Université PSL, \\  CNRS, LAMSADE, 75016 Paris, France \\
	\texttt{hugo.gilbert@lamsade.dauphine.fr} \\
\And
	\href{https://orcid.org/0009-0004-7111-3202}{\includegraphics[scale=0.06]{orcid.pdf}\hspace{1mm}Guillaume Méroué} \\
Inria, Université Côte d’Azur, CNRS, I3S, \\ 
Templier 1, 06410 Biot, France \\
	\texttt{guillaume.meroue@inria.fr} \\
	\And
	\href{https://orcid.org/0000-0001-6789-9573}{\includegraphics[scale=0.06]{orcid.pdf}\hspace{1mm}Jérôme Lang} \\
Université Paris-Dauphine, Université PSL, \\  CNRS, LAMSADE, 75016 Paris, France \\
	\texttt{jerome.lang@lamsade.dauphine.fr} \\
\And
	\href{https://orcid.org/0000-0002-1570-5244}{\includegraphics[scale=0.06]{orcid.pdf}\hspace{1mm}Sylvain Bouveret} \\
Université Grenoble-Alpes, LIG, \\ 38402 Saint-Martin-d'Hères, France
 \\
	\texttt{sylvain.bouveret@imag.fr} \\
}

\hypersetup{
pdftitle={Constrained Serial Dictatorship can be Fair},
pdfsubject={cs, cs.AI, cs.MA},
pdfauthor={Sylvain Bouveret,Hugo Gilbert, Guillaume Méroué, Jérôme Lang},
pdfkeywords={Computational social choice, Fair division},
}

\usepackage{amsmath, amsfonts, amsthm}
\usepackage{xcolor}
\usepackage{environ}
\usepackage{tikz}
\usepackage{pgfplots}
\usepackage{comment}
\pgfplotsset{compat=1.17}
\usepackage{thmtools}  
\usepackage{thm-restate}
\usepackage{algpseudocode}

\newcommand{\AgentSet}{\mathcal{A}}
\newcommand{\ItemSet}{\mathcal{G}}

\newcommand{\AllocationFunction}[2]{\pi_{#1}^{#2}}
\newcommand{\ProbModel}{\Psi}
\newcommand{\Rank}{\mathtt{rk}}
\newcommand{\FI}{\mathtt{FI}}
\newcommand{\FC}{\mathtt{FC}}

\newcommand{\IC}{\mathtt{IC}}
\newcommand{\PL}{\mathtt{PL}}
\newcommand{\Ml}{\mathtt{Mll}}
\newcommand{\eu}{\mathtt{eu}}

\newcommand{\FIC}[1]{\FI\text{-}\FC_{#1}}

\newtheorem{proposition}{Proposition}
\newtheorem{lemma}{Lemma}

\newtheorem{theorem}{Theorem}

\newtheorem{definition}{Definition}
\newtheorem{example}{Example}

\newcommand{\HG}[1]{#1}
\newcommand{\ejl}{}
\newcommand{\bjl}{} 
\newcommand{\cjl}[1]{}
\newcommand{\chg}[1]{}
\newcommand{\SB}[1]{}
\newcommand{\cgm}[1]{}

\NewEnviron{cproblem}[1]{%
\begin{center}\fbox{\parbox{2.9in}{%
    {\centering\scshape #1\par}%
    \parskip=1ex
    \BODY
}}\end{center}}

\title{Constrained Serial Dictatorships can be Fair}

\begin{document}

\maketitle

\begin{abstract}
When allocating indivisible items to agents, it is known that the only strategyproof mechanisms that satisfy a set of rather mild conditions are {\em constrained serial dictatorships}: given a fixed order over agents, at each step the designated agent chooses a given number of items (depending on her position in the sequence). 
%With these rules, 
Agents who come earlier in the sequence have a larger choice of items; however,
this advantage can be compensated by a higher number of items received by those who come later. How to balance priority in the sequence and number of items received is a nontrivial question. We use a previous model, parameterized by a mapping from ranks to scores, a social welfare functional, and a distribution over preference profiles. For several meaningful choices of parameters, we show that the optimal sequence can be computed exactly in polynomial time  or approximated using  
%by resorting to 
sampling.  
Our results hold for several probabilistic models on preference profiles, with an 
%particular 
emphasis on the Plackett-Luce model. 
We conclude with experimental results showing how the optimal sequence is impacted by various parameters.
%of the problem.  
\end{abstract}

\section{Introduction}\label{sec:intro}

In an ideal world, a mechanism for dividing a set of indivisible goods (or items, we use both terms interchangeably) 
%\cjl{Goods or items?}\HG{Well as you want, but most of the article is dedicated to goods, even if most of the results can probably be adapted to the case of bads.} 
should be at the same time efficient, fair, and insensitive to strategic behaviour. Now, strategyproofness is a very strong requirement that severely limits the choice of mechanisms. The question we address in this paper is, {\em how can we design strategyproof mechanisms while retaining an acceptable level of fairness and/or efficiency}? 

It is known that {\em under mild conditions, the only strategyproof mechanisms are within the family of serial dictatorships} (although the landscape is less dramatic when there are only two agents, see our related work section).  
A standard serial dictatorship is defined by a permutation of the set of agents; at each step, the designated agent chooses all the items she likes from those that are still available.
A {\em constrained serial dictatorship (CSD)}, also called quota serial dictatorship, is similar  except that at each step, the designated agent chooses a predefined number of items. 
%\footnote{There are also variants, called sequential dictatorships, where the choice of the next agent depends on the allocation to the previous agents in the sequence.}
 
(Constrained or unconstrained) serial dictatorships are strategyproof; but are they acceptable on efficiency and fairness grounds? Unconstrained serial dictatorships are clearly not: if the first agent likes all items then she will pick them all. Constrained serial dictatorships do better, at the price of the loss of Pareto-efficiency; but still, agents appearing early in the sequence have a much larger choice than those appearing late. This is patent in the case where there are as many items as agents, each agent being entitled to only one item, CSDs cannot do better than this: the first agent will get her preferred item, and the last agent will have no choice and might receive her least preferred item. 

However, when there are more items than agents, and agents can receive several items, things become better, because the advantage towards agents who come early in the sequence can be compensated by a higher number of items received by those who come later. Suppose, as a simple example, that three items have to be assigned to two agents, A(nn) and B(ob). Assuming that Ann picks first, there are three CSDs: AAA (Ann picks all items), AAB (Ann picks two, Bob one), and ABB (Ann picks one, Bob two). 
It is intuitively clear that ABB is optimal, but how can optimality be defined? With four items, things are less clear: AAAA and AAAB are clearly less desirable than AABB and ABBB, but which of these two should we choose? And what if we have five agents and seventeen items? 

To sum up: strategyproofness leaves us almost no choice but (constrained) serial dictatorship; some are intuitively better than others. What remains to be done is to {\em define formal optimality criteria for choosing between CSDs}, and to {\em compute optimal ones}. Our paper addresses these questions.

A way of answering the first question has been suggested by \citet{BouveretL11}, and further studied by \citet{KalinowskiNW13} (in the more general context of {\em picking sequences}).\footnote{Picking sequences are more general as agents don't necessarily pick their items in a row; for instance, the sequence where Ann picks one item, Bob two, and Ann picks the last remaining item, is not a CSD. Round Robin (perfect alternation) is another example. CSDs coincide with {\em non-interleaving} picking sequences.} 
A way of measuring the efficiency and fairness of a CSD consists in estimating the {\em expected social welfare}, according to some social welfare functional \cite{DASPREMONT2002459} -- egalitarian, Nash, utilitarian\footnote{If our main objective is {\em fairness}, 
utilitarian social welfare may not fit well. We will see further that it is the case indeed.} --  of the allocation resulting from the application of the serial dictatorship. 
Because the agents' values for items are not known, \citet{BouveretL11} estimate them from the ranks of items in an agent's preference relation: for each agent $i$, the value of item ranked in position $j$ is a fixed value $s_j$, independent from $i$. 
To estimate the expected social welfare, in addition to the non-increasing scoring vector $(s_1, \ldots, s_m)$ one also needs to assume a probability distribution over ordinal preference profiles.  These preferences can be drawn following different models as impartial culture, or more generally the Mallows~\cite{mallows1957non}, or  Plackett-Luce models~\cite{luce1959individual,plackett1975analysis}. 
%Two simple (and extreme) choices are {\em full independence} (agents' preferences over items are mutually independent) and {\em full correlation} (agents have the same preferences over items).\cjl{Shall we keep this sentence here? Hugo : We can discuss possible modifications, maybe just citing Mallows and PL models.} 

These three components (scoring vector, probability over profiles, social welfare functional) allow to associate an expected social welfare with any CSD.   
We define optimal CSDs this way, for various scoring vectors, three social welfare functionals, and various probability distributions. 

For egalitarian social welfare, we provide a simple algorithm which returns an optimal CSD given that one can compute the expected utility obtained by an agent when a CSD is used. This algorithm makes it possible to compute an optimal (respectively, close to optimal) CSD when this expected utility is polynomial-time computable (respectively, can be approximately evaluated, e.g., by sampling). 
%For utilitarian and Nash social welfare functionals, we show that computing an optimal CSD is NP-hard in the general case. \cjl{Do we really prove that?} \HG{Not really, we should remove it. } Yet, 
We also provide a dynamic programming algorithm that computes an optimal CSD for utilitarian, Nash or egalitarian social welfare under a specific condition, which is met when preferences are fully correlated, or when they are fully independent and follow the impartial culture or more generally the Plackett-Luce model. 

Sections~\ref{sec:related} and \ref{sec:model} discuss related work and present our model. 
Section~\ref{sec : computing optimal} presents our algorithms for computing an optimal CSD. 
These algorithms assume the existence of an oracle which can compute or estimate the expected utility of a picker given a CSD. 
Section~\ref{sec : utility agent} designs such oracles under various model assumptions.
%In Section~\ref{sec:dp}, we show that an optimal CSD can be computed in polynomial time given a mild condition on the distribution over profiles.
Section~\ref{sec:ic-results} gives results for small values of $n$, and depicts and comments on the evolution of the optimal sequences when all criteria except one are fixed.

\section{Related work}\label{sec:related}
%\HG{To complete ?}\cjl{I cites a few more (recent) references, plus three that I might discuss later.}

\paragraph{Strategyproof allocation of indivisible goods} 

Various characterization theorems state that, under mild additional conditions, strategyproof allocation mechanisms all have a serial dictatorship flavour:  
with strict 
%(and not necessarily monotonic) 
preferences over subsets,
\bjl 
only serial dictatorships are strategyproof, neutral, and nonbossy
\cite{Svensson99}, whereas only sequential dictatorships (a generalization of serial dictatorship where the identity of the agent picking in position $k$ depends on the items assigned to the agents in positions 1 to $k-1$) are strategyproof, Pareto-efficient, and nonbossy \cite{Papai01}.
\ejl
%The latter result still holds when replacing Pareto-efficiency by sovereignty (every allocation is the output of the mechanism for some profile). 
If preferences are quantity-monotonic (a bundle of larger cardinality is always preferred to one of lower cardinality) then a mechanism is strategyproof, nonbossy, Pareto-efficient and neutral if and only if it is a CSD (also called a quota serial dictatorship) \cite{Papai00}. Similar characterizations hold replacing quantity-monotonic by lexicographic preferences \bjl \cite{HosseiniL19,HosseiniSVX21}.\ejl With standard monotonicity, only quasi-dictatorships remain, where only the first agent in the sequence is allowed to pick more than one item \cite{Papai00}. 
Variants of these characterizations have been established by \citet{EhlersKlaus03}, \citet{BogomolnaiaDE05} and \citet{Hatfield09}. 
Ignoring Pareto-efficiency or neutrality opens the door to more complex strategyproof mechanisms; a full characterization in the  two-agent case is given by \citet{AmanatidisBCM17}. 
\bjl \citet{AmanatidisBM16} show that the CSD where all agents except the last one pick only one item  is a $\nicefrac{1}{\lfloor\frac{n-m+2}{2}\rfloor}$-approximation to maxmin fair share. 
%\citet{BuTao24} study strategyproofness together with a relaxation of envy-freeness for randomized mechanisms. 
\ejl 
%Finally, 
Weakening strategyproofness into non obvious manipulability opens the door for more possibilities \cite{PsomasV22}.
%including round robin sequences.

\citet{NguyenBR18} show that when agents have preferences over sets of items defined from preferences over single items by an extension principle, some scoring rules are strategyproof for some extension principles. Allowing randomized mechanisms offers more possibilities, but not much \bjl \cite{HosseiniL19,Kojima09,GargP22,BuTao24}.
\ejl 
%and fairness guarantees, if any, only hold {\em ex ante}. 
CSDs are also considered in chore allocation \cite{AzizLW19}.

\paragraph{Picking sequences}
Sequential allocation of indivisible goods, also known as picking sequences, originates from 
\citet{KC71}, with a game-theoretic study of the alternating sequence for two agents. 
Still for two agents, \citet{Brams00} 
%go beyond strict alternation and 
consider other particular sequences. 
\citet{BouveretL11} define a more general class of sequences, for any number of agents, and argue that sequences can be compared with respect to their expected social welfare, using a scoring vector and a prior distribution over profiles. \citet{KalinowskiNW13} show that computing the expected utility of a sequence is polynomial under full independence, and that strict alternation is optimal for two agents,  utilitarian social welfare and Borda scoring.  
The manipulation of picking sequences is studied by \citet{BL-ECAI14}, \citet{TominagaTY16} and \citet{ABLM-AAAI17}. \citet{FlamminiG20} and \citet{xiao2020algorithms} study the parameterized complexity of computing an optimal manipulation. Game-theoretic aspects of 
picking sequences are addressed by \citet{KalinowskiNWX13}.  
\citet{ChakrabortySS21} study picking sequences for agents with different entitlements. While all these works are oblivious to agent identities, \citet{DBLP:conf/sagt/CaragiannisR23} try to find an approximately optimum order of agents in a serial dictatorship with a limited number of queries. 

\paragraph{Maximizing social welfare in allocation of indivisible goods}

A classical way of guaranteeing a level of fairness and/or efficiency consists in finding an allocation {\em maximizing social welfare}, under the assumption that the input contains, for each agent, her utility function over all bundles of goods
(usually assumed additive).
%\footnote{In most works, this utility function is assumed additive.} 
Egalitarian social welfare places fairness above all, utilitarian social welfare cares only about efficiency only, and Nash social welfare is considered as a sweet spot inbetween. 
%As there are a lot of references, 
See \cite{BouveretCM16,LangR24,AmanatidisABFLMVW23,AzizLMW22}
for surveys. 
%Importantly, 
These mechanisms
%defined by maximizing social welfare from an input consisting of utility functions,
are not strategyproof.
%\cjl{Possibly cite: 
%\cite{CorenoB24}
%\cite{WangSGXCW23}.}

\section{Preliminaries : The model}\label{sec:model}

Given \(n \in \mathbb{N}^*\), we use \([n]\) to denote \(\{1,\ldots,n\}\) and \([n]_0\) to denote \(\{0,1, \ldots, n\}\). 
Bold symbols represent vectors. 

Let \(\AgentSet = \{a_1,\ldots,a_n\}\)  be a set of \(n\) agents with $a_i$ the $i^{th}$ agent to intervene in the allocation process and \(\ItemSet = \{g_1,\ldots,g_m\}\) a set of \(m\) goods. A  preference profile \(\boldsymbol{P}\! =\! (\succ_{a_1}, \ldots,\succ_{a_n})\) describes the preferences of the agents:  \(\succ_a\) is a ranking that specifies the preferences of agent \(a\) over the goods in \(\ItemSet\). We denote by \(\Rank_{\boldsymbol{P}}^a(g)\), the rank of item \(g\) in the ranking of \(a\), given profile \(\boldsymbol{P}\). 
{\em The preference profile is hidden, and therefore not part of the input:} we will assume that rankings are drawn independently according to some probabilistic model, that we denote by \(\ProbModel\). 

Two well-known probabilistic models are the \emph{Mallows} and \emph{Plackett-Luce models}~\cite{mallows1957non,luce1959individual,plackett1975analysis}: 
\begin{itemize}
    \item The \emph{Mallows model} is parameterized by a dispersion parameter $\phi\in [0,1]$ and a ranking $\mu$. We denote this model by $\Ml_{\mu,\phi}$. In this model, the probability of a ranking $r$ is proportional to $\phi^{d_{\mathtt{KT}}(r,\mu)}$, with $d_{\mathtt{KT}}(r,\mu)$, the Kendall-Tau distance between rankings $r$ and $\mu$.
    \item The \emph{Plackett-Luce (PL) model} is parameterized by a value vector $\boldsymbol{\nu} = (\nu_1, \ldots, \nu_m)$. 
    Intuitively, $\nu_i > 0$ represents the social value of good $g_i$. 
    In this model, which we denote  by $\PL_{\boldsymbol{\nu}}$, the probability of a ranking $r = g_{i_1} \succ g_{i_2} \succ ... \succ g_{i_m}$ is:
    $$\prod_{j=1}^m \frac{\nu_{i_j}}{\sum_{l = j}^m \nu_{i_l}}. $$
    The Plackett-Luce model has proven particularly good for learning a preference relation over a set of items (a.k.a. label ranking) \cite{cheng2010label} so it fits particularly well here. 
\end{itemize}
%Note that 
These models generalize the two following sub-cases:  
\begin{itemize}
    %\item The {\em full independence} condition, denoted by \(\FI\) stipulates that the rankings of the agents are sampled independently from one another, given a probability distribution over preference rankings.  
    \item \emph{Impartial Culture}, denoted by $\IC$, in which each preference ranking is drawn u.a.r. from the set of all possible rankings. Impartial culture is obtained when $\phi = 1$ for the Mallows model and when all values in $\boldsymbol{\nu}$ are equal  for the Plackett-Luce model.
    \item The {\em Full Correlation} case, denoted by \(\FC\) stipulates that all agents have exactly the same preference ranking. Full correlation is obtained when $\phi = 0$ for the Mallows model (and also as the limit of Plackett-Luce models $\boldsymbol{\nu}^M = (M^{m-1}, \ldots, M, 1)$ when $M \rightarrow \infty$).%\SB{What is $r$ exactly here? I guess it is not a ranking like before. Is it a constant number?}\HG{I think it should be $\boldsymbol{\nu}^M = (M^m, M^{m-1}, \ldots, 1)$, I let Jerome confirm. }
   % when $\boldsymbol{\nu}$ tends to a lexicographic vector of the form $\nu = ($.
    % I remove this for now \item Mixtures of the two previous models. In model \(\FIC{\lambda}\), a preference profile is sampled according to \(\FI\) with probability \(\lambda\) and \(\FC\) with probability \(1-\lambda\). 
    %Hence, parameter \(\lambda \in [0,1]\) makes it possible to adjust the level of correlation between the preferences of the agents. 
\end{itemize}

%In the first case, we will consider the following distributions: \cjl{Je suis assez circonspect quant à cette structure d'argumenration : les distributions Mallows et Plackett-Luce apportent une forme de corrélation (certes les votes restent indépendants étant donné le modèle, mais le modèle lui-même  permet d'exprimer des corrélationsfortes, puisque Full Correlation est aussi un cas particulier de Mallows et de Plackett-Luce. Je propose de réécrire comme suit : on commence par introduire Mallows et Plackett-Luce, et on donne IC et FC comme des cas particuliers. Hugo : OK je suis d'accord.}

In the sequel, we obtain different results for $\ProbModel \in \{\FC, \IC, \Ml_{\mu,\phi}, \PL_{\boldsymbol{\nu}}\}$.

The items are allocated to the different agents according to a CSD: given a vector \(\boldsymbol{k} = (k_1, \ldots, k_n)\) of \(n\) non-negative integers, agent \(a_1\) will first pick \(k_1\) goods, then \(a_2\) will pick \(k_2\) goods within the remaining ones, and so on until \(a_n\) picks \(k_n\) items. 
In most cases, we will consider complete CSDs, in the sense that \(\sum_{i=1}^n k_i = m\). 
However, we may also consider incomplete CSDs such that  \(\sum_{i=1}^n k_i < m\).  
We assume that agents behave greedily by choosing their preferred goods within the remaining ones. 
This sequential process leads to an allocation that we denote by \(\AllocationFunction{\boldsymbol{P}}{\boldsymbol{k}}\). More formally, \(\AllocationFunction{\boldsymbol{P}}{\boldsymbol{k}}\) is a function such that \(\AllocationFunction{\boldsymbol{P}}{\boldsymbol{k}}(a)\) is the set of goods that agent \(a\) has obtained at the end of the sequential allocation process, given preference profile \(\boldsymbol{P}\) and vector \(\boldsymbol{k}\).

The utility of an agent for obtaining an item \(i\) will be derived using a scoring vector. Stated otherwise, there is a vector \(\boldsymbol{s} = (s_1,\ldots, s_m)\in \mathbb{Q^+}^m\) such that  \(s_i \ge s_{i+1}\) for all \(i\in [m-1]\). %in the case of goods and \(s_i \le s_{i+1}\) in the case of bads. 
The value received by an agent for obtaining her \(j^{th}\) preferred item is \(s_j\).  Different scoring vectors can be considered. An important example
is
%: %\fjerome{Je vire quasi-indifference et $k$-approval.} 
%\begin{itemize}
    %\item T
    the {\em Borda} scoring vector, where $s_i = m-i+1$. 
    %and 
    %\item The quasi-indifference scoring vector: $ s_i = 1 + (m-i+1)/M$, where $M$ is an integer much larger than $m^2$;
    %\item T
    %the {\em lexicographic} scoring vector, where $s_i = 2^{m-i}$. 
    %\item The $k$-approval scoring vector: $s_i = 1$ if $i\le k$ and $0$ otherwise.
%\end{itemize}
Using scores as proxy for utilities is classical in social choice: this is exactly how positional scoring voting rules (e.g., the Borda rule) are defined, and they have are also used in fair division settings \cite{BramsEF03,DarmannK16,BaumeisterBLNNR17}.

We denote by \(U_{\boldsymbol{P}}^{\boldsymbol{k}}(a) = \sum_{g \in \AllocationFunction{\boldsymbol{P}}{\boldsymbol{k}}(a)} s_{\Rank_{\boldsymbol{P}}^a(g)}\) the utility 
obtained by \(a\) when receiving \(\AllocationFunction{\boldsymbol{P}}{\boldsymbol{k}}(a)\) and by \(EU_{\ProbModel}^{\boldsymbol{k}}(a) = \mathbb{E}_{\boldsymbol{P}\sim \ProbModel}[U_{\boldsymbol{P}}^{\boldsymbol{k}}(a)]\) her expected utility given model \(\ProbModel\). %Note that 
This assumes that agents have {\em additive} preferences, which is very common in fair division. 
The utilitarian social welfare (USW) \(W_\ProbModel^U(\boldsymbol{k})\),  egalitarian social welfare (ESW) \(SW_\ProbModel^E(\boldsymbol{k})\), and Nash social welfare (NSW) \bjl \(SW_\ProbModel^N(\boldsymbol{k})\) \ejl are then  defined by:\bjl 
\begin{align*}
    &SW_\ProbModel^U(\boldsymbol{k}) = \sum_{a\in \AgentSet} EU_{\ProbModel}^{\boldsymbol{k}}(a), \hspace{0.6cm}
    SW_\ProbModel^E(k) = \min_{a\in \AgentSet} EU_{\ProbModel}^{\boldsymbol{k}}(a),  \hspace{0.6cm}\\ 
    &~~~~~~~~~~~~~~~~~~~~~~~~~SW_\ProbModel^N(\boldsymbol{k}) = \prod_{a\in \AgentSet} EU_{\ProbModel}^{\boldsymbol{k}}(a).
\end{align*}
\ejl
Note that our social welfare notions are meant \emph{ex ante}, i.e., we define them on the expected utility values of the agents. This is different from the notion of \emph{ex post} social welfare which considers the utility of the agents once the profile $\boldsymbol{P}$ issued from $\ProbModel$ is determined.\\

Our objective is to study the  following class of optimization problems OptSD-\(\ProbModel\)-\(x\) with \(x\in \{U,E,N\}\).
\begin{cproblem}{OptSD-\(\ProbModel\)-\(x\)}
\textbf{Input}: A number \(n\) of agents, a number \(m\) of goods, and a scoring vector \(\boldsymbol{s}\).\\
\textbf{Find}: A vector \(\boldsymbol{k} = (k_1, \ldots, k_n)\) of \(n\) non-negative integers with $\sum_{i=1}^n k_i= m$ maximizing \(SW_\ProbModel^x(\boldsymbol{k})\).
\end{cproblem}

The following easy observation will be useful: 
\begin{restatable}{observation}{obsOne}
  For given \(n\) and \(m\), the number of vectors \(\boldsymbol{k} =  (k_1, \ldots, k_n)\) 
  such that \(\sum_{i = 1}^n k_i = m\)
  equals \(\binom{n+m-1}{n-1}\).\label{boundOnNbSol}
\end{restatable}

From this observation, we can deduce that the number of potential vectors is lower-bounded by \(\frac{m^n}{(n-1)!}\). This number does not take into account a natural further assumption that the optimal sequence is {\em non-decreasing}, that is, that $k_1 \leq k_2 \leq \ldots \leq k_n$. We will see further that this assumption holds for ESW (under a mild condition), {\em but not} for USW. When the assumption holds, we can restrict the search to non-decreasing vectors; their number is the number of integer partitions $m$ into $n$ numbers; it is still exponentially large, but no closed form expression is known.

%\HG{JE BOUGERAIS CELA A L ALGO DE PROG DYN}
%In this section, we first study the problem of computing $EU_{\ProbModel}^{\boldsymbol{k}}(a)$. 
%Because the model is anonymous, for a given distribution $\ProbModel$, $EU_{\ProbModel}^{\boldsymbol{k}}(a)$ depends only on the whole CSD vector $\boldsymbol{k}$, and the position of the agent in it (that is, for $\boldsymbol{k} = AAABBBBB$, $A$ has position 1 and $B$ has position 2).  Yet, we will see that for many interesting distributions, $EU_{\ProbModel}^{\boldsymbol{k}}(a)$ only depends on the number of items picked by $a$, and the number of items that have been picked before $a$, but not on the number of agents who have picked before and how many items they have picked each. This property will be essential for the analyses of OptSD-\(\ProbModel\)-\(x\), and will condition the existence of the dynamic programming algorithm.

%\begin{condition} \label{condition : dp}
%Let $a$ be the $i^{th}$ picker in a CSD defined by a vector $\boldsymbol{k}$. $EU_{\ProbModel}^{\boldsymbol{k}}(a)$ only depends on (1) $\kappa = k_i$, the number of goods that she picks, and (2) $\tau = \sum_{j=1}^{i-1}k_j$, the number of goods that have been picked before she starts picking.
%\end{condition} 

\section{Computing an optimal CSD} \label{sec : computing optimal}

We now investigate the problem OptSD-\(\ProbModel\)-\(x\) with \(x\in \{U, E, N\}\).
%, of computing a CSD maximizing the expected social welfare. 
\bjl All algorithms in this Section assume access \ejl to an oracle algorithm $\mathcal{T}_{\ProbModel}(\boldsymbol{k},i)$ computing $EU^{\boldsymbol{k}}_{\ProbModel}(a_i)$ in time $K(n,m, \boldsymbol{s})$. \bjl The computation of expected utilities of agents 
%in the sequence 
for various models will be addressed in Section \ref{sec : utility agent}. \ejl

We start by a positive result for Egalitarian Social Welfare: the optimal CSD can be computed by the greedy-like Algorithm~\ref{alg:esw}. 
%that works with a greedy flavor. 
%In this algorithm, 
$\mathtt{Completion}(\boldsymbol{k})$ denotes, for any partial CSD $\boldsymbol{k}$, the complete CSD such that $\mathtt{Completion}(\mathbf{k})_i \!=\! k_i$ for $i \!\in\! [n-1]$ and $\mathtt{Completion}(\mathbf{k})_n \!=\! m \!-\! \sum_{i\in [n-1]} k_i$. 
%Put differently,
In informal terms, $\mathbf{k}$ is completed by giving all remaining goods to the last agent. 

\begin{algorithm}
    \caption{GreedyESW}\label{alg:esw}
    \begin{algorithmic}[1]
        \Require the number of agents $n$, the number of goods $m$, the scoring vector $\boldsymbol{s}$, the oracle algorithm $\mathcal{T}_{\ProbModel}$ 
        \State $\boldsymbol{k} \gets (0, \ldots, 0)$ 
        %\# a vector of $n$ 0
        \# empty CSD
        \State $\mathbf{max\_k}, \max\_{esw} \gets \boldsymbol{k}, 0$
        \For{$t = 1$ to $m$}
            \State $i \in argmin_{i\in [n]} EU^{\boldsymbol{k}}_{\ProbModel}(a_i)$
            \State $k_i \gets k_i + 1$
            \If{$SW^E_{\ProbModel}(\boldsymbol{k}) > \max\_{esw}$}
                \State $\mathbf{max\_k}, \max\_{esw} \gets \boldsymbol{k}, SW^E_{\ProbModel}(\boldsymbol{k})$
            \EndIf
        \EndFor
        \State \Return $\mathtt{Completion}(\mathbf{max\_k})$;
    \end{algorithmic}
\end{algorithm}

At line 1, we start with an empty CSD,
%$\boldsymbol{k} = (0,\ldots,0)$ 
that we will modify in a greedy fashion. 
%Indeed, 
In the \textbf{for} loop (lines 3-9), we identify an agent with minimal expected utility (line 4)
%at line 4 
and increment the number of goods that she gets (line 5).
%at line 5. 
The CSD that is returned is not necessarily this CSD $\boldsymbol{k}$. 
During the algorithm, we keep in variables $\max\_{esw}$ and $\boldsymbol{\max\_k}$, the maximum ESW found so far and the corresponding (partial) CSD. 
%At line 10, 
The algorithm returns $\boldsymbol{\max\_k}$ completed by giving all remaining goods to the last agent (line 10).
The completion step is not really necessary (the partial sequence obtained at line 9 already has maximum expected egalitarian social welfare); its role is to ensure that no good is left unallocated. 
The reason why one needs the test at line  6 is that letting the currently least happy agent pick one more good may decrease the ESW, as can be seen on the following example.

\begin{example}
    Let $n = 2$, $m = 5$, $s = (50, 10, 4, 2, 1)$, and $\ProbModel = \IC$. We show below the partial CSDs obtained in each iteration $t$ 
    together with the expected utilities of both agents  
    (they can be computed easily, as we will see in Section  \ref{sec : utility agent})   
    and the values of $i_t$ and $\max\_{esw}$.
    %\footnote{A hint: at step 3,  both agents have picked one item; the expected utility of the second agent is 42 because with probability $\nicefrac15$, the first agent picked her best item and she will pick her second best, and with probability $\nicefrac45$ she will get her best item.} 
    \ejl 
    $$\begin{array}{c|cccccc}
        t & \boldsymbol{k} & \boldsymbol{\max\_k} & EU^{\boldsymbol{k}}_{\ProbModel}(a_1) 
        & EU^{\boldsymbol{k}}_{\ProbModel}(a_2)
         & i^t
         & \max\_{esw}
         \\ \hline
         1 &  (0,0) & (0,0) & 0 & 0 & 1 & 0\\
         2 &  (1,0) & (0,0) & 50 & 0 & 2 & 0\\
         3 &  (1,1) & (1,1) & 50 & 42 & 2 & 42\\
         4 &  (1,2) & (1,2) & 50 & 49.6 & 2 & 49.6\\
         5 &  (1,3) & (1,3) & 50 & 52.4 & 1 & 50\\
         6 &  (2,3) & (1,3) & 60 & 40.2 & 1 & 50
         %\\
         %7 &  (2,3) & (1,4) & 50 & 53.6 & 1 & 0
    \end{array}$$
    
    At iteration 5, the least happy agent is $a_1$; however, letting $a_1$ pick one more good, that is, $\boldsymbol{k} = (2,3)$ gives $EU^{\boldsymbol{k}}_{\ProbModel}(a_1) = 60$ and $EU^{\boldsymbol{k}}_{\ProbModel}(a_1) = 40.2$ (iteration 6), decreasing the currently optimal expected ESW. Therefore, $\boldsymbol{\max\_k}$ is not replaced by $\boldsymbol{k} = (2,3)$ at line 6
    of the algorithm. 
    The algorithm returns $\mathtt{Completion}(\mathbf{max\_k}) = (1,4)$ (with expected utilities  50 and 53.6) with the remaining good given to $a_2$.
    %(line 10). 
\end{example}

\begin{proposition}\label{thm:Algo1}
    Algorithm~\ref{alg:esw} returns a CSD $\boldsymbol{k}$ maximizing $SW^E_{\ProbModel}(\boldsymbol{k})$, solving problem OptSD-\(\ProbModel\)-\(E\), in time $O(nmK(n,m,\boldsymbol{s}))$.
\end{proposition}

The proof is based on the following lemma:

\begin{lemma}\label{lemma : algo esw}
    Let $\hat{\boldsymbol{k}}$ \bjl be a 
    %arbitrary 
    CSD.\ejl Let $\max\_{esw}^t$, $\boldsymbol{k}^t$ and $i^t$ denote $\max\_{esw}$, $\boldsymbol{k}$ and $i$ after line 4 of iteration $t$ of the \textbf{for} loop in Algorithm~\ref{alg:esw}. For all $t$, a necessary condition for $SW^E_{\ProbModel}(\hat{\boldsymbol{k}}) > \max\_{esw}^t$ is that $\hat{k}_j \ge k_j^t$ for all $j\in [n]$, and $\hat{k}_{i^t} > k^t_{i^t}$.
\end{lemma}

\begin{proof}
    By induction. At iteration 0, the claim is obvious. Assume that the claim holds for iteration $t$, and let $\hat{\boldsymbol{k}}$ be a CSD such that $SW^E_{\ProbModel}(\hat{\boldsymbol{k}}) > \max\_{esw}^{t+1}$. Then obviously $SW^E_{\ProbModel}(\hat{\boldsymbol{k}}) > \max\_{esw}^{t}$ as $\max\_{esw}^{t+1} \ge \max\_{esw}^{t}$.
    Because the condition holds for iteration $t$ and by construction of $\boldsymbol{k}^{t+1}$ we have that $\hat{k}_j \ge k^{t+1}_j$ for all $j\in [n]$.
    Now suppose that $\hat{k}_{i^{t+1}} = k_{i^{t+1}}^{t+1}$. In that case, $SW^E_{\ProbModel}(\hat{\boldsymbol{k}}) \leq EU^{\boldsymbol{k}^{t+1}}(a_{i^{t+1}}) \le \max\_{esw}^{t+1}$, a contradiction with the induction hypothesis. The first inequality is due to the fact that $a_{i^{t+1}}$ will get the same number of goods in $\hat{\boldsymbol{k}}$ and $\boldsymbol{k}^{t+1}$ while the agents picking before her will get at least as many goods in $\hat{\boldsymbol{k}}$ than in  $\boldsymbol{k}^{t+1}$. The second inequality is due to the definition of $i^{t+1}$.
\end{proof}

% \begin{proof}[Proof of Proposition~\ref{thm:Algo1}]
%     Consider towards a contradiction a CSD $\hat{\boldsymbol{k}}$ such that $SW^E_{\ProbModel}(\hat{\boldsymbol{k}}) > \max\_{esw}$. Lemma~\ref{lemma : algo esw} applied at iteration $t = m$ implies that $\hat{\boldsymbol{k}}$ is equal to the greedily constructed complete CSD $\boldsymbol{k}$ obtained at the end of the \textbf{for} loop. Yet, $\max\_{esw} \ge SW^E_{\ProbModel}(\boldsymbol{k})$ a contradiction.
% \end{proof}

\begin{proof}[Proof of Proposition~\ref{thm:Algo1}] %By contradiction.
     Suppose that there exists a CSD $\hat{\boldsymbol{k}}$ such that $SW^E_{\ProbModel}(\hat{\boldsymbol{k}}) > \max\_{esw}$.
Lemma~\ref{lemma : algo esw} applied at iteration $t = m$ implies that each agent receives more objects with $\hat{\boldsymbol{k}}$ than with the greedily constructed complete CSD $\boldsymbol{k}$ obtained at the end of the \textbf{for} loop. As they both have $m$ objects to allocate, they must be equal. This is a contradiction of the hypothesis as $\max\_{esw} \ge SW^E_{\ProbModel}(\boldsymbol{k})$.
\end{proof}

%Note that 
%\HG{BOUGER A LA PROCHAINE SECTION}
%Propositions~\ref{prop : FC}, \ref{prop : IC} and \ref{thm:Algo1} imply that the problem OptSD-\(\ProbModel\)-\(E\) can be  solved in polynomial time for $\ProbModel = \FC$ and $\ProbModel=\IC$. We obtain a complexity of \(O(m\max(n,m))\) for $\ProbModel = \FC$ and \(O(m\max(n,m^2))\) for $\ProbModel = \IC$ by precomputing all values $\eu(\kappa,\tau)$ before running the  algorithm.

We now go beyond egalitarian social welfare. For utilitarian and Nash social welfare, we do not know of an efficient algorithm which would work for any distribution.
%\cjl{Well, brute force search over all sequences plus estimating expected utilities by sampling would do, right?}\HG{I added ``efficient''.} 
A general approach could be to sample a large but hopefully reasonable number of preference profiles from $\ProbModel$ and find a CSD with maximal social welfare considering the average utility of each agent. Yet, we prove in  Appendix~\ref{app : sec 5} that such an approach leads to an NP-hard problem for USW.

However, provided the distribution satisfies a natural condition, a CSD maximizing utilitarian and Nash social welfare can be computed by dynamic programming. This condition on $\ProbModel$ states that $EU_{\ProbModel}^{\boldsymbol{k}}(a)$ only depends on the number of items picked by $a$, and the number of items that have been picked before $a$, but not on the number of agents who have picked before and how many items they have picked each.
%neither on $a$'s identity, nor on the identity of the agents that have picked before). 
%This property will be essential for the analyses of OptSD-\(\ProbModel\)-\(x\), and will condition the existence of the dynamic programming algorithm.
%\HG{Should we use prefix independence in lower case to mention the condition of Prefix Independence or PI ?}

\bjl
\begin{definition}
A distribution $\ProbModel$ satisfies {\bf prefix independence} if for any  sequence $\boldsymbol{k}$ and $i \in [n]$,  if 
%\begin{condition}[Prefix independence] \label{condition : dp} Let 
$a$ is the $i^{th}$ picker in  $\boldsymbol{k}$, then
%a CSD defined by a vector $\boldsymbol{k}$. 
$EU_{\ProbModel}^{\boldsymbol{k}}(a)$ only depends on (1) $\kappa = k_i$, the number of goods that she picks, and (2) $\tau = \sum_{j=1}^{i-1}k_j$, the number of goods that have been picked before she starts picking.
%The expected utility that an agent gets in the allocation process depends only
%should only depend 
%on (1) the number of goods that she picks and (2) the number of goods that have been picked before she starts picking. 
\end{definition} 

\ejl
\HG{
Under} \bjl     prefix independence\ejl, 
\HG{the utility that agent $a$ gets when picking $\kappa$ goods while $\tau$ have already been picked, $\eu(\kappa,\tau)$, is well-defined, and is exactly equal to $EU_{\ProbModel}^{\boldsymbol{k}}(a)$ when $a$ is the $i^{th}$ picker $\kappa \! = \! k_i$ and $\tau \!=\! \sum_{j=1}^{i-1}k_j$.
}

%We now show that, when \bjl prefix independence \ejlholds, a dynamic programming algorithm makes it possible to solve OptSD-\(\ProbModel\)-\(x\) for $x\in\{\HG{E,U,N}\}$. %\cjl{Should we still mention $E$?\HG{I would say yes}}
\cjl{I removed ``We now show that, when \bjl prefix independence \ejl
holds, a dynamic programming algorithm makes it possible to solve OptSD-\(\ProbModel\)-\(x\) for $x\in\{\HG{E,U,N}\}$.''
}
For pedagogical purposes, let us first focus on %the case where we wish to 
maximising USW. 
When \HG{prefix independence} is met, one can use the following 
%easy 
dynamic programming equations:
\begin{align}
    F(i,\tau) =& \max_{\kappa \in [m - \tau]_0} (\eu(\kappa,\tau) \!+\! F(i+1,\tau+\kappa)), \notag\\
    &\hspace{2cm}\forall i,\tau \in [n-1]\times [m]_0  \label{eq : DP1 utilitarian},\\
    F(n,\tau) =& \eu(m-\tau,\tau), \forall \tau \in [m]_0, \notag
\end{align}
where \(F(i,\tau)\) corresponds to the maximum USW that can be obtained by agents $\{a_i,a_{i+1},\ldots,a_n\}$ in the situation in which $\tau$ goods have already been allocated and we allocate the $m-\tau$ remaining goods to them.
%\begin{itemize}
%    \item \(F(i,\tau)\) corresponds to the maximum USW that can be obtained by agents $\{a_i,a_{i+1},\ldots,a_n\}$ in the situation in which $\tau$ goods have already been allocated and we allocate the $m-\tau$ remaining goods to them;
    %\item \(\eu(k,t)\) denotes the expected utility obtained by an agent if we allocate \(k\) items to her knowing that \(t\) items have been allocated to the preceding agents. % preceding her in the allocation process.
%\end{itemize}
Of course the optimal value is given by \(F(1,0)\). 

The other problems can be solved similarly. For \HG{problem OptSD-\(\ProbModel\)-\(E\)} (resp. \bjl OptSD-\(\ProbModel\)-\(N\)), \ejl one should adapt Equation~\ref{eq : DP1 utilitarian} by replacing the sum operation between $\eu(\kappa,\tau)$ and $F(i+1,\tau+\kappa)$ by a min (resp. multiplication) operation. %replace Equation~\ref{eq : DP1 utilitarian} by Equation~\ref{eq : DP1 egalitarian} (resp. Equation \ref{eq : DP1 nash}).\SB{These equations could probably be removed if we need more space.}\HG{True.}
%\begin{align}
%F(i,\tau) &= \max_{\kappa\in [m-\tau]_0} (\min(\eu(\kappa,\tau),F(i+1,\tau + \kappa))), \notag\\
%    &\hspace{2cm}\forall i,\tau \in [n-1]\times [m]_0; \label{eq : DP1 egalitarian}\\
%F(i,\tau) &= \max_{k\in [m-\tau]_0} (\eu(\kappa,\tau)\times F(i+1,\tau + \kappa)), \notag\\
%    &\hspace{2cm}\forall i,\tau \in [n-1]\times [m]_0. \label{eq : DP1 nash}
%\end{align}

\begin{proposition}\label{thm:Algo2}
 If $\Psi$ satisfies \bjl prefix independence,
%Condition \ref{condition : dp}, 
\ejl problems OptSD-\(\ProbModel\)-\(U\), OptSD-\(\ProbModel\)-\(E\) and \bjl OptSD-\(\ProbModel\)-\(N\) \ejl can be solved in \(O(nm^2K(n,m,\boldsymbol{s}))\) time.
\end{proposition}

We conclude by giving a structural property %which is 
satisfied by an optimal CSD \HG{for ESW} when  \bjl prefix independence holds.\ejl We will see that such property does not necessarily hold for USW \HG{(see Appendix~\ref{app : sec 5} and Section~\ref{sec:ic-results})}. 
\begin{restatable}{proposition}{propEgNI}\label{prop : egalitarian is non-increasing}
Under \bjl prefix independence, \ejl there exists an optimal solution to OptSD-\(\ProbModel\)-\(E\) which is non-decreasing, i.e., in which the earlier an agent picks, the less goods she gets. 
\end{restatable}

\section{Computing the expected utility of an agent}\label{sec : utility agent}

In this section, we address the computation of $EU_{\ProbModel}^{\boldsymbol{k}}(a)$.
\bjl
%In this section, we first study the problem of computing $EU_{\ProbModel}^{\boldsymbol{k}}(a)$. 
%Because the model is anonymous, for a given distribution $\ProbModel$, $EU_{\ProbModel}^{\boldsymbol{k}}(a)$ depends only on the whole CSD vector $\boldsymbol{k}$, and the position of the agent in it (that is, for $\boldsymbol{k} = AAABBBBB$, $A$ has position 1 and $B$ has position 2).
\bjl Prefix independence again plays a crucial role: when it is satisfied, 
%Yet, we will see that for many interesting distributions, 
$EU_{\ProbModel}^{\boldsymbol{k}}(a)$ only depends on the number of items picked by $a$, and the number of items that have been picked before $a$, but not on the number of agents who have picked before and how many items they have picked each. 
We first investigate which of our different probabilistic models satisfy it.
%starting with $\FC$ and $\IC$. \ejl

%However, in most cases the computation of $EU_{\ProbModel}^{\boldsymbol{k}}(a)$ will reveal easier when the following condition holds. \cjl{Right but we never say anything more precise than that. The reader expects something like ``When condition 1 is met then the expected utility of an agent can be computed exactly in polynomial time''. This is not true in general, because we could imagine weird models where Condition1 is met but expected utilities are hard to compute. However, for the models we consider, we have more than Condition 1: we have a closed formula, which is polynomial-time computable. Perhaps the main consequence of Condition 1 is not that computing an agent' expected utility is easy (not that it could be easy without Condition 1) but that maxiizing utilitarian and Nash social welfare will be polynomial-time computable via DP. Right?}\HG{“When condition 1 is met then the expected utility of an agent can be computed exactly in polynomial time”. We don’t know if this is true for the PL model as we have no closed formula.  Perhaps the main consequence of Condition 1 is not that computing an agent’ expected utility is easy (not that it could be easy without Condition 1) but that maximizing utilitarian and Nash social welfare will be polynomial-time computable via DP. Right?
%No because it just says the DP works not that it is P-time.}

%Let us warm up with $\FC$ and $\IC$:
\bjl
\begin{proposition}
%Condition~\ref{condition : dp} is met when 
$\ProbModel \in \{\FC, \IC\}$ satisfy prefix independence.
\end{proposition}
\ejl
%\cjl{For pedagogical reasons? Given that we said above that FC and FI were special cases of PL, this is not really needed anymore.} \HG{PL only yields FC in the limit so a proof is still needed. For FI I would still keep it as it helps understanding how to compute the utility on agent exactly.} 

\begin{proof}
Consider a situation where an agent starts picking while $\tau$ goods have previously been picked. When \(\ProbModel = \FC\) or \(\ProbModel = \IC\), the probability distribution on the set $S$ of goods that have previously been picked only depends on $\tau$: for \(\ProbModel = \FC\), this probability distribution assigns probability 1 to the set composed of the $\tau$ (unanimously) most preferred goods; for \(\ProbModel = \IC\), this probability distribution assigns equal probability to all sets of size \(\tau\) and 0 to others. 
Note that, given the set $S$, the utility that the agents get is then determined by the number of goods she picks.
\end{proof}

More interestingly, the $\PL_{\nu}$ model, which generalizes $\FC$ and $\IC$, also satisfies \bjl prefix independence. \ejl 
%Condition~\ref{condition : dp}.

\bjl
\begin{proposition}\label{proposition : PL condition 1}
%Condition~\ref{condition : dp} is met when 
$\ProbModel = \PL_{\boldsymbol{\nu}}$ satisfies prefix independence.
\end{proposition}
\ejl

To reason on the Plackett-Luce model, one can use the \emph{vase model metaphor} \bjl \cite{silverberg1980statistical} \ejl
%proposed by \citet{silverberg1980statistical}. 
Consider a vase filled with $m$ types of balls, \bjl 
%such that 
the proportion of balls of type $j$ being \ejl $f(j) = \frac{\nu_{j}}{\sum_{l = 1}^m \nu_{l}}$\cjl{I removed ``(the number of balls contained in the vase is infinite, allowing the desired proportions)''}
%(the number of balls contained in the vase is infinite, allowing the desired proportions). 
The ranking is then generated by the following sequential process. 
At each stage, a ball is taken from the vase such that a ball of type $j$ is chosen with probability $f(j)$. 
If the ball is of a different type than the ones previously picked, it yields the next good in the ranking. In either case, the ball is put back in the vase and the process continues. 
Using this metaphor, one can prove the following lemma (whose formal proof is postponed to Appendix~\ref{app : sec 4}). 
%\HG{conflit de notation sur $s$ qui est aussi le scoring vector.}
\begin{restatable}{lemma}{lemmaPLTwo}\label{lemma : PL 2}
    \bjl Let $I = (i_1, \ldots, i_q)$ be a sequence of $q$ different indices in $[m]$. 
    Consider the following two cases:
    \begin{enumerate}
        \item[i)] Agent $a_1$ picks $q$ goods;
        \item[ii)] Agent $a_1$ picks $q_1$ goods and agent $a_2$ picks $q_2$ goods with $q_1 + q_2 = q$.
    \end{enumerate}
    For the PL model, the probability that for all $t\in [q]$,  \ejl $g_{i_t}$ is picked at timestep $t$ is the same in cases i and ii. 
\end{restatable}

\begin{proof}[Proof of Proposition~\ref{proposition : PL condition 1}]
    We recall that the preference rankings of the agents are drawn independently from $\PL_\nu$. 
    Using Lemma~\ref{lemma : PL 2} and a simple induction argument, we get that the probability of a specific sequence of \bjl $q$ \ejl consecutive picks is the same regardless of whether they were picked by one, two or more agents. This entails that the probability distribution on the set $S$ of goods that have been picked after $\tau$ timesteps only depends on the value of $\tau$. Hence, the expected utility that an agent gets when choosing $\kappa$ goods once $\tau$ have been picked only depends on the values of $\kappa$ and $\tau$.
\end{proof}

%\HG{Here a paragraph is missing on the computation of the utility of an agent in the PL model by sampling. }

Unfortunately, things are different for the Mallows model:

\begin{restatable}{proposition}{propMallows}
   There exists $\phi\in(0,1)$ and a ranking $\mu$ such that %Condition~\ref{condition : dp} is not met with 
   $\ProbModel = \Ml_{\phi,\mu}$ \bjl does not satisfy prefix independence.\ejl
\end{restatable}

This holds even for 3 agents and 3 goods. See Appendix~\ref{app : sec 4} for the proof. 
%counterexample is detailed in Appendix.

\paragraph{Computation of $EU_{\ProbModel}^{\boldsymbol{k}}(a)$}

%\cjl{The notation is confusing: we have  $k$ as the number of goods to be picked by the agent and $\boldsymbol{k}$ as the picking vector. When both are subscripts they are so small that it's very hard to make the difference. Par ailleurs j'ai du mal à comprendre formellement l'égalité $\eu(k,t) = EU_{\ProbModel}^{\boldsymbol{k}}(a)$ vu que les arguments ne sont pas les mêmes. Bien sûr, je comprends l'idée. Par ailleurs je propose de supprimer ``Condition~\ref{condition : dp} states that a probability distribution exhibits structural properties that guarantee the independence of $EU_{\ProbModel}^{\boldsymbol{k}}(a)$ to the  full identity of picking agents. '' qui est un peu redondant et nous fera gagne run peu de place.}\HG{J'ai essayé de clarifier et je suis d'accord.}
%Condition~\ref{condition : dp} states that a probability distribution exhibits structural properties that guarantee the independence of $EU_{\ProbModel}^{\boldsymbol{k}}(a)$ to the  full identity of picking agents. 
%However, it does not say anything on the possibility to compute $EU_{\ProbModel}^{\boldsymbol{k}}(a)$ efficiently. 
\bjl Under prefix independence, we %now 
show how to compute $\eu(\kappa,\tau)$ efficiently,
%the utility that an agent gets when picking $k$ goods while $t$ have already been picked (which is exactly equal to $EU_{\ProbModel}^{\boldsymbol{k}}(a)$ 
%(when Condition~\ref{condition : dp} holds),  
starting by $\FC$.\ejl
%the case $\ProbModel \!=\! \FC$.

\begin{proposition} \label{prop : FC}
If $\ProbModel = \FC$, 
$\eu(\kappa,\tau) = \sum_{i = \tau+1}^{\tau+\kappa} s_i$. All values $\eu(\kappa,\tau)$ can be computed in time $O(m^2)$ with the recursive formula $\eu(\kappa,\tau) = \eu(\kappa-1,\tau) + s_{\kappa+\tau}$. %\footnote{Complexities are given omitting the   logarithmic terms induced by the arithmetic operations in the complexity formula.} 
%\cjl{I don't understand Footnote 6. Is it needed anyway? We don't care about logarithmic terms when we have a quadratic running time.}\HG{Agreed I removed it. }
\end{proposition}

We then turn to $\ProbModel = \IC$, and show that the values $\eu(\kappa,\tau)$ can be computed using a recursive formula. 
Let $T(j,\kappa,\tau)$ denote the utility that an agent can get if she can pick $\kappa$ goods within the ones of rank in $\{j,\ldots, m\}$, given that $\tau$ of these goods have been picked by preceding agents. Then, it is clear that we have:
\begin{equation*}
\eu(\kappa,\tau) = T(1,\kappa,\tau), \forall \kappa,\tau\in [m]_0\times[m-\kappa]_0    
\end{equation*}
The key point is that there is a probability $1 - \frac{\tau}{m - j+1}$ that the good of rank $j$ is free and in this case the agent will pick this good, and a probability of $\frac{\tau}{m - j+1}$ that this good is one of the $\tau$ goods that have previously been picked. In both cases, we move to goods of rank in $\{j+1,\ldots, m\}$. In the first case, we decrease $\kappa$ by one as the agent has picked a good. In the second case, we decrease $\tau$ by 1 as we have identified one of the goods picked within the ones of rank $j$ to $m$. Hence, $\eu(\kappa,\tau)$ can be computed by the following formula:%\footnote{This formula has some analogy with Equation 2 in the paper by \citet{KalinowskiNW13}.}
%\cjl{I removed the footnote: This formula has some analogy with Equation 2 in the paper by \citet{KalinowskiNW13}.}
\begin{align}
T(j,\kappa,\tau) &= (1 - \frac{\tau}{m - j+1})(s_j + T(j+1, \kappa-1, \tau)) \notag\\
+ &\frac{\tau}{m \!-\! j\!+\!1}T(j\!+\!1, \kappa, \tau\!-\!1),\notag\\ 
\forall j,\kappa,\tau& \in [m-1]\!\!\times\!\![m-j+1]_0\!\!\times\!\![m-j-\kappa+1],\label{eq : rec formula}
\end{align}
with the following base cases:
\begin{align*}
T(j,0,\tau) &= 0, \forall j,\tau\in[m]\times[m-j+1]_0\\
T(j,\kappa,0) &= \sum_{j \le i < j+\kappa} s_i,\forall j,\kappa\in[m]\times[m-j+1]_0.
\end{align*}
 By computing all values \(T(j,\kappa,\tau)\) in \(O(m^3)\) operations, we obtain the following result.
\begin{proposition} \label{prop : IC}
If $\ProbModel = \IC$, then all values \(\eu(\kappa,\tau)\) can be computed in time \(O(m^3)\) by using Equation~\ref{eq : rec formula}.   
\end{proposition}

\bjl
 Propositions~\ref{thm:Algo2}, \ref{prop : FC}, and \ref{prop : IC} imply that 
 %the problem 
 OptSD-\(\ProbModel\)-\(x\) for $x\in \{U, E, N\}$ can be  solved in polynomial time for $\ProbModel = \FC$ and $\ProbModel=\IC$, in
 %We \bjl get \ejl a complexity 
 \(O(nm^2)\) for $\ProbModel = \FC$ and \(O(m^2\max(n,m))\) for $\ProbModel = \IC$, \ejl by precomputing all values $\eu(\kappa,\tau)$ before running the dynamic programming algorithm.

For $\ProbModel\not\in\{\IC,\FC\}$, one can still use GreedyESW and the dynamic programming algorithm with values $EU^{\boldsymbol{k}}_{\ProbModel}(a)$ approximated by sampling, providing close-to optimal CSDs:
\bjl 
%More precisely, 
the returned CSD is optimal 
%for the social welfare functional 
with expected utility values 
%of the agents 
\ejl
replaced by their approximate values.\footnote{Some mild monotonicity conditions are required on the approximated $EU^{\boldsymbol{k}}_{\ProbModel}(a)$ values for the validity of Algorithm~\ref{alg:esw}.}
\ejl

For the general $\PL_{\nu}$ model beyond $\FC$ and $\IC$, we do not know whether values $\eu(\kappa, \tau)$ can be computed exactly in polynomial time; 
%efficiently; 
however, they %what counts is that 
can be efficiently approximated  
%in practice, we can still efficiently compute an approximation of this quantity 
by sampling preference profiles from $\ProbModel$ and averaging the utility values obtained on the samples, %Interestingly, n that case, 
with approximation guarantees from Hoeffding's (\citeyear{hoeffding1963probability}) inequality. 

To present this guarantee, let  $u_{\kappa,\tau}(\boldsymbol{P},\boldsymbol{s})$ denote the utility value obtained by the second picker when she picks her $\kappa$ preferred (available) goods, while the first picker has picked her $\tau$ preferred ones, given the preference profile $\boldsymbol{P}$.
%\cjl{yes}
%\HG{Similarly, a paragraph is missing here, on computation by sampling}
%\HG{Conflit de notation avec U ici, non juste U majuscule c'est déjà Utilisé avec un autre sens}\cjl{Parce qu'on note les utilit\'es par $u$? Mais alors c'est depuis le d\'ebut...}
\begin{restatable}{proposition}{samplingBound}\label{samplingBound}
    Let $\epsilon \!>\! 0$ and $\delta \!\in\! (0,1)$ two fixed values, and $\Upsilon$ 
    %\SB{Not a big fan of notation $\overline{U}$ for an upper bound (that reminds me more like an average or expected utility). Why not simply $U$?} \HG{Agreed !}
    an upper bound on values $\eu(\kappa,\tau)$ (e.g., $\sum_{i=1}^m s_i$). 
    
    Let $\widetilde{\eu}_{\kappa,\tau}$ be the value computed by averaging the values $u_{\kappa,\tau}(\boldsymbol{P}_i,\boldsymbol{s})$  over $N$ preference profiles $\boldsymbol{P}_i$ sampled independently from $\ProbModel$. 
    %\footnote{I don't understand ``obtained over $(U^2\ln{(2m^2/\delta)})/2\epsilon^2$''. For someone who's not well-acquainted with Hoeffding's bounds -- which is my case -- this is hardly understandable.}
    %\HG{I tried making it clearer}
    If $N \ge (\Upsilon^2\ln{(2m^2/\delta)})/2\epsilon^2$, then it holds with probability $1-\delta$ that:
    $$|\eu(\kappa,\tau) - \widetilde{\eu}_{\kappa,\tau}|\le \epsilon, \forall \kappa,\tau\in [m]\times[m-\kappa].$$
\end{restatable}

\HG{Moreover, we show that these utility values can be computed exactly in time FPT (Fixed-Parameter Tractable) with respect to parameter $m$ and XP (slicewise polynomial) with respect to $\rho$, where $\rho$ is the number of distinct values in $\nu$. 
This seems  particularly appealing as goods may often be partitioned  in categories. When $\rho = 1$, all goods are in the same category and we obtain the $\IC$ model; when $\rho$ equals 2 or 3 we obtain categories $\{$high value, low value$\}$ or $\{$high value, medium value, low value$\}$.   
\begin{restatable}{proposition}{propPLm} \label{prop : PL m}
If $\ProbModel = \PL_{\nu}$, then all values \(\eu(\kappa,\tau)\) can be computed in time \(O(4^mPoly(m))\).   
\end{restatable}
%A possibly further interesting result is that the problem is in XP if we consider the parameter  Indeed, each set $S\subset G$ can then be replaced by a vector $(m_1,\ldots,m_{\rho})$ with the number of goods of each type in $S$ (bounded by $m^\rho$). An algorithm based on recursive formulas as the one presented before can then be used. 
\begin{restatable}{proposition}{propPLrho} \label{prop : PL rho}
If $\ProbModel = \PL_{\nu}$, then all values \(\eu(\kappa,\tau)\) can be computed in time \(O(m^{2\rho}Poly(m))\).   
\end{restatable}}

\section{Numerical tests}\label{sec:ic-results}
We performed several experiments to explore the properties of the CSDs obtained by maximizing either USW, NSW or ESW. More precisely, we explored the impact of increasing one of the parameters, %(e.g., number of agents or items, or distribution of preferences) 
all other parameters being fixed. %These investigations were led using several scoring vectors. 
%\cjl{We could add a table that gives the optimal CSDs (exact or approximate) for, e.g., $n = 2, 3, 4$, a few selected values of $m \in \{4, \ldots, 10\}$, $\Psi$ = $FC, FI$, one selected $PL$, and one selected Mallows, our three social welfare functionals, Borda scoring (plus possible other vectors but the table is quite big already!). If we select three values of $m$, that makes $3 \times 3 \times 4 \times 3$, that makes already 36 sequences to give. (not sure how it fits in a single table). What's important is that we have at least one instance for which agent 1 gets more goods than some other agent for USW. We'll gain space by removing some of the comments below about the impact of $\lambda$.}

\paragraph{Impact of the number of goods} 
Figure~\ref{fig:results with borda} displays  the proportion of utility (left-hand side) and goods (right-hand side) obtained for $n=5$ and increasing the number of goods $m$ from 5 to 300 in steps of 5. To generate both figures, the $\IC$ model and the Borda scoring vector were used and we optimized either USW, ESW or NSW.\footnote{In Figures~\ref{fig:results with borda}, and \ref{fig : Mll and Luce}, the 1st picker corresponds to the color {\em blue} (at the bottom of each plot) while the 5th and last agent to pick corresponds to the color {\em purple} (at the top of each plot). Moreover, note that the values plotted are in fact cumulative values.}

%\MG{Hugo m'a parlé de frais suplémentaires pour les couleurs, c'est toujour le cas ?}\jerome{Jamais entendu parler de ça, mais si c'est le cas ce sera seulement pour la version finale. Ce qui est important, cependant, cest que les figures restent lisibles même en noir et blanc, pour les relecteurs qui impriment le papier plutôt que le lire sur écran.}
%\include{plots} \\
\begin{figure}[t]
    \centering 
\begin{tikzpicture}[scale=0.65]
	\begin{axis}[
	scale = 0.8,
		xlabel = {m},
		ylabel = {\small Portion of total utility},
	const plot,
		stack plots=y,
		area style,
				ymin=0,
        axis y line = left,
        ylabel style={yshift=-0.5em},
		enlarge x limits=false]
		\foreach \n in{0,...,4}{
    \addplot table[x=m, y=U/Su\n] {./Data/5_300_Borda_utilitarian_Tout.txt}
    \closedcycle;
}
	\end{axis}
\end{tikzpicture}$~~~$\begin{tikzpicture}[scale=0.65]
	\begin{axis}[
	scale = 0.8,
		xlabel = {m},
		ylabel = {\small Portion of goods received},
	const plot,
		stack plots=y,
		ymin=0,
		area style,
        axis y line = right,
        ylabel style={yshift=0.5em},
		enlarge x limits=false,
		]
		\foreach \n in{0,...,4}{
    \addplot table[x=m, y=Nb/m\n] {./Data/5_300_Borda_utilitarian_Tout.txt}
    \closedcycle;
}
	\end{axis}
\end{tikzpicture}
    
    \begin{tikzpicture}[scale=0.65]
	\begin{axis}[
	xlabel = {m},
	ylabel = {\small Portion of total utility},
	scale = 0.8,
	const plot,
		stack plots=y,
		area style,
				ymin=0,
        axis y line = left,
        ylabel style={yshift=-0.5em},
		enlarge x limits=false]
		\foreach \n in{0,...,4}{
    \addplot table[x=m, y=U/Su\n] {./Data/5_300_Borda_egalitarian_Tout.txt}
    \closedcycle;
}
	\end{axis}
\end{tikzpicture}$~~~$\begin{tikzpicture}[scale=0.65]
	\begin{axis}[
	scale = 0.8,
		xlabel = {m},
	ylabel = {\small Portion of goods received},
	const plot,
		stack plots=y,
		ymin=0,
		area style,
        axis y line = right,
        ylabel style={yshift=0.5em},
		enlarge x limits=false,
		]
		\foreach \n in{0,...,4}{
    \addplot table[x=m, y=Nb/m\n] {./Data/5_300_Borda_egalitarian_Tout.txt}
    \closedcycle;
}
	\end{axis}
\end{tikzpicture}

\begin{tikzpicture}[scale=0.65]
	\begin{axis}[
	scale = 0.8,
	xlabel = {m},
	ylabel = {\small Portion of total utility},
	const plot,
		stack plots=y,
		area style,
				ymin=0,
        axis y line = left,
        ylabel style={yshift=-0.5em},
		enlarge x limits=false]
		\foreach \n in{0,...,4}{
    \addplot table[x=m, y=U/Su\n] {./Data/5_300_Borda_nash_Tout.txt}
    \closedcycle;
}
	\end{axis}
\end{tikzpicture}$~~~$\begin{tikzpicture}[scale=0.65]
	\begin{axis}[
	scale = 0.8,
		xlabel = {m},
		ylabel = {\small Portion of goods received},
	const plot,
		stack plots=y,
		ymin=0,
		area style,
        axis y line = right,
        ylabel style={yshift=0.5em},
		enlarge x limits=false,
		]
		\foreach \n in{0,...,4}{
    \addplot table[x=m, y=Nb/m\n] {./Data/5_300_Borda_nash_Tout.txt}
    \closedcycle;
}
	\end{axis}
\end{tikzpicture}
    \caption{\small Portion of total utility (plots on the left) 
    and of goods (right) received by each of 5 agents with $m$ increasing from 5 to 300 in steps of 5.
    %results are obtained with an optimal CSD when maximizing either the 
    Maximizing USW (plots at the top), NSW (bottom), or ESW (middle), using Borda scoring vector and $\IC$. }
    \label{fig:results with borda}
\end{figure}
%\begin{itemize}
%\item 
    
Several comments can be made. First, as expected, in the egalitarian case (middle of Figure~\ref{fig:results with borda}), we observe that as $m$ increases, the distribution of utility received by each agent converges towards equal share.\footnote{This observation is proven formally in Appendix~\ref{app : sec 6}.}
%the one in which they all receive the same share of utility. 
In order to achieve this, the agents who arrive later in the sequence receive more items. 

Second, 
%we observe that 
with %
%the combination of 
Borda and utilitarianism, the first agent 
in the sequence
may pick more items than others (plots on top of Figure~\ref{fig:results with borda}). More generally, on this plot, the utility of an agent seems to decrease with the position in the %picking 
    sequence.  %This finding can seem at first counter-intuitive. 
    
Finally, for the Borda scoring vector, egalitarian and Nash social welfare objectives tend to give similar results. 
    %For other scoring vectors the three notions of social welfare often lead to similar outcomes. 
    %\item \HG{What can be said for plots in Appendix ?} For the lexicographic scoring vector, the egalitarian and utilitarian social welfare notions seem to lead to similar results.
%\end{itemize}

 \paragraph{Impact of correlation} 
 We explore the impact of 
 %the 
 correlation, through the parameters $\phi$ and $\boldsymbol{\nu}$ of models $\PL_{\boldsymbol{\nu}}$ and $\Ml_{\phi,\mu}$. We use the Borda scoring vector and maximize ESW. To run Algorithm~\ref{alg:esw}, we approximate the expected utility values of the agents by sampling $10000$ preference profiles from $\PL_{\boldsymbol{\nu}}$ and from $\Ml_{\phi,\mu}$ with the PrefSampling library \cite{boehmer2024guide}.
 % from the distribution. 
 Figure~\ref{fig : Mll and Luce} displays the utility value (plots at the bottom) and the number of items (top) received by each of 5 agents with $m = 70$ goods, for models $\PL_{\boldsymbol{\nu}}$ (right) and $\Ml_{\phi,\mu}$ (left). In the former model, we use $\boldsymbol{\nu}^x = (x^m,x^{m-1},\ldots,x^1)$ and decrease $x$ from $1.5$ (which already yields very correlated preference profiles in 
 %the spirit of 
 similar to 
 $\FC$) to $1$ ($\IC$) in steps of $0.01$. In the latter model, we increase $\phi$ from 0 ($\FC$) to 1 ($\IC$) in steps of $0.02$. 
 
Several comments are in order. 
First, as can be seen in Figure~\ref{fig : Mll and Luce}, the utility values of all agents (and hence their sum) increase when $x$ decreases or $\phi$ increases. Indeed, as we come closer to $\IC$, the preferences of the agents become more different, allowing some agents to receive some of their preferred items even if they pick late in the allocation process. 
    
Second, the number of goods received by the first agents in the CSD increases while it decreases for the last ones. Indeed, as these latter agents can receive more preferred goods, the CSD needs less to compensate by giving them a high number of goods (recall that we optimize ESW). 

Third, we notice that both models $\PL_{\boldsymbol{\nu}}$ and $\Ml_{\phi,\mu}$ yield very similar plots as we decrease the level of correlation. 

A web application can be found in the supplementary material which makes it possible to explore the  characteristics of optimal CSDs for various scoring vectors, probabilistic models, number of agents and goods.

\begin{figure}[!ht]
    \centering 
\begin{tikzpicture}[scale=0.65]
	\begin{axis}[
		xlabel = {$\phi$},
		ylabel = {\small Number of items received},
        scale = 0.8,
		const plot,
		stack plots=y,
		area style,
		ymin=0,
		axis y line = left,
		ylabel style={yshift=-0.2em},
		enlarge x limits=false]
		\foreach \n in {0,...,4}{
			\addplot table[x=phi, y=Nb\n] {./Data/mallows_phi_allocation.txt}
			\closedcycle;
		}
	\end{axis}
\end{tikzpicture}$~$\begin{tikzpicture}[scale=0.65]
	\begin{axis}[
		xlabel = {$x$},
		ylabel = {\small Number of items received},
        scale = 0.8,
		const plot,
		stack plots=y,
		area style,
		ymin=0,
		axis y line = right,
		ylabel style={yshift=-0.2em},
		enlarge x limits=false,
		x dir=reverse]
		\foreach \n in {0,...,4}{
			\addplot table[x=x, y=Nb\n] {./Data/luce_x_allocation.txt}
			\closedcycle;
		}
	\end{axis}
\end{tikzpicture}

\begin{tikzpicture}[scale=0.65]
	\begin{axis}[
		xlabel = {$\phi$},
		ylabel = {\small Utility Value ($\times 10^3$)},
        scale = 0.8,
		const plot,
		stack plots=y,
		ymin=0,
        yticklabels={0,0, 1,2,3},
		area style,
		axis y line = left,
		ylabel style={yshift=-0.2em},
		enlarge x limits=false]
		\foreach \n in {0,...,4}{
			\addplot table[x=phi, y=U\n] {./Data/mallows_phi_utilities.txt}
			\closedcycle;
		}
	\end{axis}
\end{tikzpicture}\begin{tikzpicture}[scale=0.65]
	\begin{axis}[
		xlabel = {$x$},
		ylabel = {\small Utility Value ($\times 10^3$)},
        scale = 0.8,
		const plot,
		stack plots=y,
		ymin=0,
        yticklabels={0,0, 1,2,3},
		area style,
		axis y line = right,
		ylabel style={yshift=0.5em},
		enlarge x limits=false,
		x dir=reverse]
		\foreach \n in {0,...,4}{
			\addplot table[x=x, y=U\n] {./Data/luce_x_utilities.txt}
			\closedcycle;
		}
	\end{axis}
\end{tikzpicture}
%\begin{tikzpicture}[scale=0.65]
%	\begin{axis}[
%		xlabel = {$\phi$},
%		ylabel = {\small Normalized Utility Value},
%		const plot,
%		stack plots=y,
%		ymin=0,
%		area style,
%		axis y line = right,
%		ylabel style={yshift=0.5em},
%		enlarge x limits=false]
%		\foreach \n in {0,...,4}{
%			\addplot table[x=phi, y=U\n] {./Data/mallows_phi_utilities_norm.txt}
%			\closedcycle;
%		}
%	\end{axis}
%\end{tikzpicture}
\caption{Number of goods received per agent (top); \bjl expected utility value per agent (bottom) as a function of %the dispersion 
$\phi$ for 
%model 
$\Ml_{\phi,\mu}$ and 
%value 
$x$ for 
%model 
$\PL_{\boldsymbol{\nu}^x}$. Maximizing ESW, Borda scoring vector, $n=5$, $m=70$.\ejl}\label{fig : Mll and Luce}
\end{figure}

%\begin{figure}[!ht]
%    \centering 
%\begin{tikzpicture}[scale=0.65]
%	\begin{axis}[
%		xlabel = {$x$},
%		ylabel = {\small Number of items received},
%		const plot,
%		stack plots=y,
%		area style,
%		ymin=0,
%		axis y line = left,
%		ylabel style={yshift=-0.2em},
%		enlarge x limits=false,
%		x dir=reverse]
%		\foreach \n in {0,...,4}{
%			\addplot table[x=x, y=Nb\n] {./Data/luce_x_allocation.txt}
%			\closedcycle;
%		}
%	\end{axis}
%\end{tikzpicture} $~~~$ 
%\begin{tikzpicture}[scale=0.65]
%	\begin{axis}[
%		xlabel = {$x$},
%		ylabel = {\small Normalized Utility Value},
%		const plot,
%		stack plots=y,
%		ymin=0,
%		area style,
%		axis y line = right,
%		ylabel style={yshift=0.5em},
%		enlarge x limits=false,
%		x dir=reverse]
%		\foreach \n in {0,...,4}{
%			\addplot table[x=x, y=U\n] {./Data%/luce_x_utilities.txt}
%			\closedcycle;
%		}
%	\end{axis}
%\end{tikzpicture}$~~~$ 
%\begin{tikzpicture}[scale=0.65]
%	\begin{axis}[
%		xlabel = {$x$},
%		ylabel = {\small Normalized Utility Value},
%		const plot,
%		stack plots=y,
%		ymin=0,
%		area style,
%		axis y line = right,
%		ylabel style={yshift=0.5em},
%		enlarge x limits=false,
%		x dir=reverse]
%		\foreach \n in {0,...,4}{
%			\addplot table[x=x, y=U\n] {./Data%/luce_x_utilities_norm.txt}
%			\closedcycle;
%		}
%	\end{axis}
%\end{tikzpicture}
%\caption{Luce qui varie avec $x$} 
%\end{figure}

%\paragraph{A note on computation times} 

%The PL model was considered in the preference learning community by \cite{guiver2009bayesian} and \cite{cheng2010label}.

%\cjl{I am moving the next two sections to the Appendix, first for speeding compilation.}

\section{Discussion}

The practical use of our setting raises a few questions.% which we discuss now.

First, we need to choose a distribution. 
%This nontrivial question cannot be answered in an absolute way: 
The choice has to be tailored to the domain at hand, and distributions can be learnt using some preference learning models and techniques. If computation time is an important issue\cjl{I removed ``because we have many agents and/or items"} 
then it is wise to learn a Plackett-Luce model \bjl \cite{cheng2010label}.
%which satisfies prefix independence
 \ejl 

Second, we need to choose a scoring vector as a proxy for agents' valuations over items. Again, this depends on the specific domain at hand. For each  context, the scores can be estimated by an experiment where subjects are presented with a list of items to elicit their valuations; see Appendix~\ref{app : expe}.
%which can be seen arbitrary (although scores as proxies for utilities are used in positional scoring voting rules).  However, for each

Third, we need to choose a social welfare functional.
%as we found out, utilitarianism can be quite unfair and should be used only with care; egalitarianism is much better, especially if we have a lot of objects, because then the utilities of all agents tend to be equal; Nash is a sweet spot inbetween but is hard to compute of the distribution does not satisfy Condition 1. 
We have seen that, unsurprisingly, utilitarianism may lead to clearly unfair solutions and should be used only with care. As usual, egalitarianism may lead to a loss of efficiency, but is easier to compute or approximate; Nash is a good trade-off (see \cite{CaragiannisKMPS19} for a manifesto towards using Nash social welfare in fair division) but is hard to compute if the distribution does not satisfy \bjl prefix independence.\ejl 
%Condition~\ref{condition : dp}. 
%, in a centralized setting). 

Four, once a CSD is found, it is anonymous: for instance, with two agents, if the output is $(1,2)$, it does not say who should start picking. Assigning agents to positions in the sequence has no impact on {\em ex ante} social welfare, but it may have an impact on {\em ex post} social welfare (see Appendix~\ref{app : AAP}).

\section{Conclusion}

Our main messages are:
\begin{enumerate}
    \item Imposing strategyproofness does not leave much choice beyond constrained serial dictatorships.
    \item Some constrained serial dictatorships are fairer than others; it is fairer to let agents coming late in the picking sequence pick more items.
    \item Their efficiency and fairness can be measured by expected social welfare, defined by a scoring vector, a distribution over profiles, and a social welfare functional.
    \item Depending on the social welfare functional notion and the distribution over profiles, the optimal sequence can be:
    \begin{itemize}
        \item polynomial-time computable,
        \item efficiently approximated by sampling,
        \item or hard to approximate by sampling.
    \end{itemize}
\end{enumerate}
\ejl The following table summarizes the results obtained for the three social welfare functionals under different distributions. 
\bjl PI means that prefix independence is satisfied, \ejl
%``Condition~\ref{condition : dp} is satisfied'', 
poly means ``polynomial-time computable'', and approx means ``efficiently approximable by sampling''.
\begin{center}
\begin{tabular}{c|ccccc}
$\ProbModel$ & PI & $EU^{\boldsymbol{k}}_{\ProbModel}(a_i)$ & Egal & Nash & Uti\\ \hline
$\FC$ & yes & poly & poly& poly& poly\\
$\IC$ & yes & poly& poly& poly& poly\\
$\PL_{\boldsymbol{\nu}}$ & yes & approx & approx & approx & approx \\
$\Ml_{\phi,\mu}$ & no & approx & approx & ? & ?
\end{tabular}
\end{center}

%Our last message has an engineering flavour: in order to choose an optimal picking sequence,  the scoring vector and, to some extent, the distribution %)  can be chosen from data, as in Section \ref{sec:expe}. 

%The choice of the social welfare functional  is more fundamental. 

For the sake of the exposition, we focused on goods. %i.e., items with positive utilities. 
If items were bads (e.g., chores) instead of goods,
\bjl 
%(e.g., chores), 
a similar methodology would work, with values in the scoring vector representing costs.  
%instead of utilities. 
\ejl 
Of course, agents coming first in the sequence should now take {\em more} items than those coming later. 

%Note that all results presented in this section can be adapted to the situation in which all items are bads. In this case, the values in the scoring vectors represent costs, and the goal is to find the vector $k$ minimizing the social cost. The same dynamic programming algorithm can then be used by replacing $\max$ operators by $\min$ ones and vice-versa. 

%\end{document}

\newpage

\bibliographystyle{unsrtnat}
\bibliography{biblio}

%\theendnotes
%\end{document}

~

\newpage

\appendix

\begin{center}
    {\Large \bf Appendix}
\end{center}

\section{Omitted Proofs of Section 3}\label{app : sec 3}

\obsOne*
\begin{proof}
Choosing $n$ numbers \((k_1, \ldots, k_n)\) matching the definition amounts to partition %the interval 
  \([m]\) into \(n\) subintervals, which in turn amounts to choose \(n-1\) ``separation bars''. Said otherwise, this comes down to choose \(n-1\) increasing numbers \(l_1 \leq \ldots \leq l_{n - 1}\) among \(m + 1\). Here, \(k_i = l_i - l_{i-1}\), with the convention that \(l_0 = 0\). This problem can be equivalently formulated as the one of drawing \(n-1\) \emph{different} numbers among \(n+m-1\). For each such draw, we can obtain a set of increasing numbers \(l'_1 \leq \ldots \leq l'_{n - 1}\) between \(1\) and \(n+m-1\), that can be cast to increasing numbers between \(0\) and \(m\) by choosing \(l_i = l'_i - i\). Since there are \(\binom{n+m-1}{n-1}\) subsets of \(n-1\) elements among \(n+m-1\), we obtain the result. 
\end{proof}
\section{Omitted Proofs of Section 4} \label{app : sec 5}

We now show that 
prefix independence 
entails a specific property for the optimal solutions of OptSD-\(\ProbModel\)-\(E\). %The proof of this result can be found in the Appendix.  

\propEgNI*
\begin{proof}
Let \(\eu(\kappa,\tau)\) denote the utility obtained by an agent if we allocate \(\kappa\) items to her knowing that \(\tau\) items have already been allocated. 
As items have positive valuations, it is easy to prove that \(\eu(\kappa,\tau)\) is non-decreasing in \(\kappa\) and non-increasing in \(\tau\). 
Let us consider a solution \(\boldsymbol{k} = (k_1,\ldots, k_n)\), which is not a non-decreasing vector. Then, there exists \(i\in [n-1]\) such that \(k_i > k_{i+1}\). We set \(\tau_i = \sum_{j=1}^{i-1} k_j\), \(\tau_{i+1} = t_i + k_i\) and \(\tau_{i+1}' = t_i + k_{i+1}\). From the properties of function $\eu$, it is clear that \(\min(\eu(k_i,\tau_i),\eu(k_{i+1},\tau_{i+1})) = \eu(k_{i+1},\tau_{i+1}) \le \min(\eu(k_{i+1},\tau_i),\eu(k_{i},\tau_{i+1}'))\). Hence, by swapping \(k_i\) and \(k_{i+1}\) in \(\boldsymbol{k}\), we do not decrease the egalitarian score of $\boldsymbol{k}$ (note that this swap does not affect the utility values received by agents other than \(a_i\) and \(a_{i+1}\)). The repetition of this argument shows that there exists an optimal solution to OptSD-\(\ProbModel\)-\(E\) which is a non-decreasing vector.
\end{proof}

We will see in the following example that this property %that an optimal sequence is non-decreasing 
fails for utilitarian social welfare. For Nash social welfare, we conjecture it holds, but so far we do not have a proof.

\begin{example}\label{ex IC}
Let \(n = 3\), \(m = 7\) and the Borda scoring vector. By dynamic programming we find the values \(\eu(\kappa,\tau)\) displayed on Table~\ref{tab:utilities in ex}.
\begin{table}[!ht]
    \caption{Utilities $\eu(\kappa,\tau)$ in Example~\ref{ex IC}; $m = 7$, $\IC$ model.}
    \centering
    \scalebox{1}{\begin{tabular}{c||c|c|c|c|c|c|c|c}
    $\kappa \backslash \tau$ &  0 & 1 & 2 & 3 & 4 & 5 & 6 & 7\\
    \hline
    0     &  0 & 0 & 0 & 0 & 0 & 0 & 0 & 0\\ 
    1     &  7 & 6.86 & 6.67 & 6.4 & 6 & 5.33 & 4 & -\\ 
    2     & 13 & 12.57 & 12 & 11.2 & 10 & 8 & - & - \\ 
    3     & 18 & 17.14 & 16 & 14.4 & 12 & - & - & - \\ 
    4     & 22 & 20.57 & 18.67 & 16 & - & - & - & - \\ 
    5     & 25 & 22.86 & 20 & - & - & - & - & -\\ 
    6     & 27 & 24 & - & - & - & - & - & -\\ 
    7     & 28 & - & - & - & - & - & - & - 
\end{tabular}}
    
    \label{tab:utilities in ex}
\end{table}
For USW, we obtain 
%by dynamic programming 
an optimal vector $\boldsymbol{k} = (3,2,2)$, yielding expected social welfare 37.2. For maximizing ESW and NSW, we obtain $\boldsymbol{k} = (2,2,3)$, with expected social welfare 12 and 1872 respectively. Note that the optimal vectors for ESW and NSW may be different: with \(n=4\) and \(m=10\) and the Borda scoring vector,  %\(k = (3,3,2,2)\) is optimal for utilitarian social welfare; 
\(\boldsymbol{k} = (2,2,2,4)\) is optimal for ESW ; and \(\boldsymbol{k} = (2,2,3,3)\) for NSW.
%Similar results can be obtained with \(n=4\) agents, \(m=10\) items and the Borda scoring vector. In that case, the optimal vectors are: 
%\begin{itemize}
%    \item \(k = (3,3,2,2)\) for maximizing the utilitarian social welfare yielding a value of  \(78.15\);
%    \item \(k = (2,2,2,4)\) for maximizing the egalitarian social welfare yielding a value of \(17.29\); 
%    \item and \(k = (2,2,3,3)\) for maximizing the Nash social welfare yielding a value of \(135476.79\).
%\end{itemize}
%In that case, we see that all three optimal solutions are different. 
\end{example}

We see on Example \ref{ex IC} that the optimal sequence for utilitarian social welfare, IC, and Borda scoring, is not non-decreasing, and thus clearly not fair: the first agent in the sequence not only has a larger choice of items but picks one more than the other two! It is not new that utilitarianism may clash fairness when looking for optimal CSDs: for instance, it is known that the optimal sequence for utilitarianism, Borda scoring, FI, $n = 2$ and $m$ even is perfect alternation $1212\ldots 12$, which is obviously not fair \cite{KalinowskiNW13}. (Still, we continue to include utilitarianism in our study, first for the sake of comparison, and second because utilitarianism is relevant in some situations.) 

\

The following result concerns utilitarian social welfare:

\begin{proposition}\label{pro:utilitarianIsHard}
    Given a scoring vector $s$, a finite set $\mathcal{P}$ of $n$-agent preference profiles over a set of goods and an integer $K$, the problem of determining whether there is a CSD $\boldsymbol{k}$ such that the average utilitarian social welfare of $\boldsymbol{k}$ over all profiles of $\mathcal{P}$ is greater than or equal to $K$ is NP-complete.
\end{proposition}

%\cjl{[IS THIS CORRECT?] We consider the following decision problem: given a scoring vector $s$, a finite set of $n$-agent preference profiles $P_1, \ldots, P_q$ over a set of goods, and a number $\theta$, is there a picking sequence $\pi$ such that the average utilitarian social welfare of $\pi$ over all profiles is at least $\theta$, $\frac{1}{q} \sum_{i = 1}^n EU(\pi,P_j) \geq \theta$? Then this problem is NP-complete, even if $s = (1, 0, \ldots, 0)$.}\HG{Yes this is correct. I havn't had time to rewrite things well for this section.}

\begin{proof}
    We will prove the proposition by reduction from Exact-Cover-By-3-Sets (X3C):
    \begin{cproblem}{X3C}
        \textbf{Input:} A set $\mathcal{X} = \{x_1,\ldots, x_n\}$ of $n$ elements; a collection $\mathcal{S} = \{S_1,\ldots, S_m\}$ of $m$ subsets such that $\forall S\in \mathcal{S}$, $S$ contains exactly three elements of $\mathcal{X}$.\\
        \textbf{Question:} Does there exist a subcollection $\mathcal{C} \subseteq \mathcal{S}$, such that $\bigcup_{S\in \mathcal{C}} S = \mathcal{X}$ and $S\cap S' = \emptyset, \forall S,S'\in \mathcal{C}$.
    \end{cproblem}

    Let $(\mathcal{X}, \mathcal{S})$ be an X3C instance. From that instance, we create an instance of our problem with $n$ goods and $3m$ agents such that the set of goods is exactly $\mathcal{X}$ (by notation abuse), and such that there are 3 agents $a_S^1$, $a_S^2$, and $a_S^3$ for each set $S \in \mathcal{S}$.

    For each set $S = \{a,b,c\}\in \mathcal{S}$, we create 3 rankings $r_S^a$, $r_S^b$, $r_S^c$ such that $r_S^a$ starts with $a$, $r_S^b$ starts with $b$, and $r_S^c$ starts with $c$ (the rest of the ranking does not matter). Then for each pair of agents $(a_S^i, a_T^j)$ with $S \neq T$, we create 9 profiles as follows:

    \begin{itemize}
    \item if $S = \{a,b,c\}$, agents $a_S^1, a_S^2, a_S^3$ have rankings from $\{(r_S^a, r_S^b, r_S^c), (r_S^b, r_S^c, r_S^a), (r_S^c, r_S^a, r_S^b)\}$;
    \item if $T = \{d,e,f\}$, agents $a_T^1, a_T^2, a_T^3$ have rankings from $\{(r_T^d, r_T^e, r_T^f), (r_T^e, r_T^f, r_T^d), (r_T^f, r_T^d, r_T^e)\}$;
    \item for each $X = \{x, y, z\}$ different from $S$ and $T$, agents $a_X^1, a_X^2, a_X^3$ have rankings $(r_X^{x}, r_X^{y}, r_X^{z})$.
    \end{itemize}

    This thus makes $\frac{81m(m-1)}{2}$ profiles in total. Now, the scoring vector is such that the top object has utility 1 while all the other items have utility 0. We will prove that there exists an exact cover iff there exists a CSD with average utility at least $n$.
    
    $(\Rightarrow)$ If $\mathcal{C} \subseteq \mathcal{S}$ is an exact cover, let agents $a_S^i$ for $S\in \mathcal{C}$ and $i \in \{1,2,3\}$ pick one item. By construction, for each profile, these agents will all have their top choices yielding a utility of $n$.
    
    $(\Leftarrow)$ Conversely, let $\boldsymbol{k}$ be a CSD yielding utility $n$ for all profiles. Necessarily $\boldsymbol{k}$ gives exactly one item to $n$ agents. Let $a_S^i$ and $a_T^j$ be two such agents, with $S \neq T$. Suppose that $S \cap T \neq \emptyset$ and let $x \in S \cap T$. In the profile where $a_S^i$ has ranking $r_S^x$ and $a_T^j$ has ranking $r_T^x$ both agents have the same top object. Hence, the utility yielded by $\boldsymbol{k}$ is necessarily strictly lower than $n$ for this profile, a contradiction with the hypothesis. This proves that all the agents $a_S^i$ and $a_T^j$ receiving an object in $\boldsymbol{k}$ are such that either $S = T$ or $ S \cap T = \emptyset$. Hence, we necessarily obtain $n/3$ sets who are pairwise disjoint and hence provide an exact cover.

\end{proof}

\section{Omitted Proofs of Section 5} \label{app : sec 4}

To prove Lemma~\ref{lemma : PL 2}, we first need to prove the following Lemma. 
\begin{lemma}\label{lemma : incomplete ranking}
    Let $\tilde{r}: g_{i_1} \succ g_{i_2} \succ \ldots \succ g_{i_q}$ be an incomplete ranking over $\ItemSet$ with $q\le m$. Under the PL model, the probability to generate a ranking $r$ which is a consistent extension of $\tilde{r}$  is equal to:
    %\cjl{By ``probability of an incomplete ranking $r: i_1 \succ i_2 \succ \ldots \succ i_q$ I guess you mean the probability that the complete $\succ$ ranking sampled by the PL model is such that $i_1 \succ i_2 \succ \ldots \succ i_q$.} \HG{ Hugo : Yes I tried to clarify it.}
    $$\prod_{j=1}^q \frac{\nu_{i_j}}{\sum_{l = j}^q \nu_{i_l}}$$
\end{lemma}
\begin{proof}
Let $S\in \ItemSet$ be the set of goods on which $\tilde{r}$ express preferences. 
The lemma can easily be derived from the vase model metaphor. Consider the following slightly different sequential process. At each stage, a ball is taken from the vase such that a ball of type $j$ is chosen with probability $f(j)$. If the ball is of a different type than the ones previously picked, \emph{and is a ball of a type corresponding to an element of $S$}, then it yields the next item in the ranking. In either case, the ball is put back in the vase and the process continues. This process generates a ranking on the elements of $S$ according to the original PL model. We get that the probability of $\tilde{r}$ is: 
$$\prod_{j=1}^q \frac{\nu_{i_j}}{\sum_{l = j}^q \nu_{i_l}}.$$
\end{proof}

\lemmaPLTwo*
\begin{proof}
    We will show that in both cases, the probability that for all $t\in [q]$ $g_{i_t}$ is picked at timestep $t$ is:
    $$ p_I = \prod_{j=1}^q \frac{\nu_{i_j}}{\sum_{l = j}^q \nu_{i_l} + \sum_{p\in [m]\setminus I} \nu_p}. $$
    
    In case i), $g_{i_{t}}$ is picked at timestep $t$ for all $t\in [q]$ if $\Rank_{\boldsymbol{P}}^{a_1}(g_{i_{t}}) = t$  for all $t\in [q]$. Under the PL model, this occurs with probability $p_I$. 
    
    In case ii), let $I_1$ (resp. $I_2$) be the subsequence composed of the $q_1$ first (resp. $q_2$ last) elements of $I$ and $S_1 = \{g_{i_1}, \ldots, g_{i_{q_1}}\}$ (resp. $S_2 = \{g_{i_{q_1+1}}, \ldots, g_{i_q}\}$).  
    %Moreover, let $I_1$ (resp. $I_2$) be the sets with the $s_1$ first (resp. $s_2$ last) indices in $I$ according to $\sigma$, and $S_1 = \{g_i : i \in I_1\}$ (resp. $S_2 = \{g_i : i \in I_2\}$).
    In the PL model, the probability that $a_1$ picks $g_{i_t}$ at timestep $t$ for all $t\in [q_1]$ is:
    $$p_I^1=\prod_{j=1}^{q_1} \frac{\nu_{i_{j}}}{\sum_{l = j}^q \nu_{i_{l}} + \sum_{p\in [m]\setminus I} \nu_p}.$$
    
    Then, the probability that $a_2$ picks $g_{i_{t}}$ at timestep $t$ for all $t\in [q]\setminus [q_1]$ corresponds to the probability that $g_{i_{t}}$ is ranked at position $t-q_1$, when restricting ourselves to the goods in $\ItemSet \setminus S_1$. 
    Put another way, goods in $S_2$ should be ranked in the top $q_2$ positions in the partial ranking which only ranks goods in $\ItemSet \setminus S_1$. 
    Let $\tilde{r} : g_{i'_1} \succ g_{i'_2} \succ \ldots \succ g_{i'_{m-q_1}}$ be one such ranking over $\ItemSet \setminus S_1$ with $g_{i'_l} = g_{i_{(q_1+l)}}$ for all $l \in [q_2]$. 
    Resorting to Lemma~\ref{lemma : incomplete ranking}, $\tilde{r}$ occurs with probability:
    $$\prod_{j=1}^{q_2} \frac{\nu_{i'_j}}{\sum_{l = j}^{m-q_1} \nu_{i'_l}}\prod_{j=q_2+1}^{m-q_1} \frac{\nu_{i'_j}}{\sum_{l = j}^{m-q_1} \nu_{i'_l}}.$$    
    By marginalizing over all such rankings, the second product vanishes as we obtain the sum of probabilities over all rankings over $\ItemSet\setminus(S_1\cup S_2)$ under $\PL_{\nu}$. Hence, we get probability:
    $$p_I^2 = \prod_{j=1}^{q_2} \frac{\nu_{i'_j}}{\sum_{l = j}^{q_2} \nu_{i'_l} + \sum_{p\in [m]\setminus I} \nu_p}.$$   
    The product of $p_I^1$ and $p_I^2$ yields exactly $p_I$.
\end{proof}

\propMallows*
\begin{proof}
To see why, consider the case of $m=3$ items and $n=3$ agents, and a Mallows model with center $\mu = a \succ b \succ c$ and parameter $\phi$. The probability of a ranking $r$ to occur is $\phi^{d_{KT}(r, \mu)}/C$ where $C = 1 + 2\phi + 2\phi^2 + \phi^3$ is a normalization constant. The probabilities of the different rankings are described on the table below.

\begin{center}
    \begin{tabular}{c|c|c}
       & Probability  &  Ranking \\\hline
       1& $1/C$ & $a \succ b \succ c$ \\
      2&  $\phi/C$ & $a \succ c \succ b$ \\
      3&  $\phi/C$ & $b \succ a \succ c$ \\
      4& $\phi^2/C$  & $b \succ c \succ a$ \\
      5& $\phi^2/C$  & $c \succ a \succ b$ \\
      6& $\phi^3/C$  & $c \succ b \succ a$
    \end{tabular}
\end{center}

    If the expected utility that an agent gets in the allocation process only depend on the number of goods that she picks and that have been picked before she started picking, then it should be the same for $a_3$ if (1) $a_1$ and $a_2$ both picked one and if (2) $a_1$ picked two and $a_2$ picked zero items. We will see that it is not the case. Let us assume that $\boldsymbol{s} = (1,1,0)$ so that $a_3$ only cares about not getting her least preferred item. 
    The probability that $a$ (resp. $b$, $c$) is her least preferred good is $\phi^2(1+\phi)/C$ (resp. $\phi(1+\phi)/C$, $(1+\phi)/C$).
    
    We can compute the probability that the first two goods picked (by agent 1 and 2) are $a$ and $b$ in the two cases. In the first case, either agent 1 picks $a$ first (so it has ranking 1 or 2) and agent 2 picks $b$ (so it has ranking 1, 3 or 4), or agent 1 picks $b$ first (ranking 3 or 4) and agent $2$ picks $a$ (ranking 1 or 2 or 3). The probability of this is

    \begin{align*}
   & \frac{1+\phi}{C} \frac{1+\phi+\phi^2}{C} + \frac{\phi+\phi^2}{C} \frac{1+\phi+\phi}{C}  \\
    =& \frac{(1+\phi)(1 + 2\phi + 3\phi^2)}{C^2} \\
    \end{align*}

    In the second case, the probability that agent 1 picks both $a$ and $b$ is the probability of rankings 1 and 3, which is:

    \begin{align*}
        \frac{1+\phi}{C} &= \frac{(1+\phi)(1+2\phi+2\phi^2+\phi^3)}{C^2}. 
    \end{align*}
    
    Similarly: 
    \begin{itemize}
    \item The probability that items $a$ and $c$ are picked by the first two agents is $\phi(1+\phi)/C$ in cases 1 and 2. 
    \item The probability that items $b$ and $c$ are picked by the first two agents is $\phi^3(1+\phi)(3+2\phi+\phi^2)/C^2$ in case 1 and $\phi^2(1+\phi)/C$ in case 2.
    \end{itemize}
    As a result, the probability that $a_3$ gets her least preferred good is:
    \begin{align*}
        \frac{(1+\phi)^2}{C^2}\frac{(1+2\phi+3\phi^2)}{C} &+ \phi^2 \frac{(1+\phi)^2}{C^2} \\
        +& \phi^4 \frac{(1+\phi)^2}{C^2}\frac{(3\phi+2\phi^2+\phi^3)}{C}
    \end{align*}
    in case 1 and
    \begin{align*}
        \frac{(1+\phi)^2}{C^2} + \phi^2 \frac{(1+\phi)^2}{C^2} + \phi^4 \frac{(1+\phi)^2}{C^2}
    \end{align*}
    in case 2. These are two different values, e.g., for $\phi = 0.5$ we obtain 0.440 in case 1 and 0.429 in case 2. 
\end{proof}

\samplingBound*
\begin{proof}

Let $\eu_{\kappa,\tau}^i$ represent the utility of an agent receiving $t$ items after $k$ items have already been taken, for the preference profile $P_i$ such that $\widetilde{\eu}_{\kappa,\tau} = \sum_{i=1}^n \eu_{\kappa,\tau}^i$.

We aim to show that with probability at least $1 - \delta$, the sampled value of utility is close to the expected utility within $\epsilon$. Formally, we want:
\[ \Pr\left(\left|\frac{\widetilde{\eu}_{\kappa,\tau}}{n} - \eu_{\kappa,\tau}\right| \leq \epsilon\right) \geq 1 - \delta \]

However, this form is not directly suitable for applying Hoeffding's inequality, so we first perform some manipulations. We have:
\[
\begin{aligned}
\Pr\left(\left|\frac{\widetilde{\eu}_{\kappa,\tau}}{n} - \eu_{\kappa,\tau}\right| \leq \epsilon\right) &= 1 - \Pr\left(\left|\frac{\widetilde{\eu}_{\kappa,\tau}}{n} - \eu_{\kappa,\tau}\right| > \epsilon\right) \\
&\geq 1 - \Pr\left(\left|\frac{\widetilde{\eu}_{\kappa,\tau}}{n} - \eu_{\kappa,\tau}\right| \geq \epsilon\right)
\end{aligned}
\]

Hoeffding's inequality gives us:
\[ \Pr\left(\left|\frac{\widetilde{\eu}_{\kappa,\tau}}{n} - \eu_{\kappa,\tau}\right| \geq \epsilon\right) \leq 2 \exp\left(-\frac{2n^2\epsilon^2}{\sum_{i=1}^n (b_i - a_i)^2}\right) \]

Given that $\eu_{\kappa,\tau}^i$ are identically distributed, we can rewrite $\sum_{i=1}^n (b_i - a_i)^2$ as $n \times (b - a)^2$, where $a = \sum_{i=1}^t s(i)$ when the agent receives the items they like the least, and $b = \sum_{i=1}^t s(m-i)$ when the agent receives the items they prefer the most. For simplicity, we set $a = 0$ (no selected item) and $b = \Upsilon$ (all items are selected).

Substituting these expressions into the inequality, we get:
\[ \Pr\left(\left|\frac{\widetilde{\eu}_{\kappa,\tau}}{n} - \eu_{\kappa,\tau}\right| \leq \epsilon\right) \geq 1 - 2 \exp\left(-\frac{2n\epsilon^2}{\Upsilon^2}\right) \]

We want this probability to be at least $1 - \delta$, which gives us:
\[ 1 - 2 \exp\left(-\frac{2n\epsilon^2}{\Upsilon^2}\right) \geq 1 - \delta \]

This rearranges to:
\[ \exp\left(\frac{2n\epsilon^2}{\Upsilon^2}\right) \geq \frac{2}{\delta} \]

% \[ \frac{2n\epsilon^2}{U^2} \geq \ln\left(\frac{2}{\delta}\right) \]

Thus, for the probability to be at least $1 - \delta$, the number of samples $n$ must satisfy the inequality below.

\[ n \geq \frac{\Upsilon^2 \ln\left(\frac{2}{\delta}\right)}{2\epsilon^2} \]

Then let $E_{\kappa,\tau}$ denote the event that $\widetilde{\eu}_{\kappa,\tau}$ is an $\epsilon$-additive approximation of $\eu_{\kappa,\tau}$. This event occurs with probability $1 - \delta$.

We want to determine the probability that all events $E_{\kappa,\tau}$ hold simultaneously. This is equivalent to computing:
\[ \Pr\left(\cap_{\kappa,\tau} E_{\kappa,\tau}\right) = 1 - \Pr\left(\cup_{\kappa,\tau} \overline{E_{\kappa,\tau}}\right) \]

Applying the union bound, we get:
\[ \Pr\left(\cap_{\kappa,\tau} E_{\kappa,\tau}\right) \geq 1 - \sum_{k,t \in [m]} \Pr\left(\overline{E_{\kappa,\tau}}\right) \]

Given that $\Pr\left(\overline{E_{\kappa,\tau}}\right) \leq \delta$, we have:
\[ \Pr\left(\cap_{\kappa,\tau} E_{\kappa,\tau}\right) \geq 1 - m^2 \delta \]

For this probability to be at least $1 - \Delta$, we require:
\[ 1 - m^2 \delta \geq 1 - \Delta \]

This rearranges to:
\[ \delta \leq \frac{\Delta}{m^2} \]

Substituting this condition into the sample size inequality, we find that $n$ must satisfy:
\[ n \geq \frac{\Upsilon^2 \ln\left(\frac{2m^2}{\Delta}\right)}{2\epsilon^2} \]

\end{proof}
\propPLm* 
\begin{proof}
Let $S\subseteq G$ be a set of goods and $g\in S$, we denote by $\mathtt{ft}(g,S) = \nu_g/(\sum_{g'\in S} \nu_{g'})$, the probability that $g$ is ranked first among the elements of $S$ according to $\PL_{\nu}$.

The proof relies on recursive equations. 
Let us consider the following setting.
The picker under consideration should  pick $\kappa$ goods within a set $S$ of goods occupying the $|S|$ last positions of her ranking. Goods in $G\setminus S$, occupy the top $m-S$ ranks in her ranking and have already been picked, either by her or by other agents. 
Moreover, the set $S'\subseteq S$ have already been picked by previous pickers. 
We are interested in computing the expected utility $U(\kappa,S,S')$ of the $\kappa$ picks of the agent in such a situation. 
We argue that $U(\kappa,S,S')$ satisfies the following recursive equations:
\begin{align}
U(\kappa, S,S’) &= \sum_{g\in S\cap{S’}} \ \mathtt{ft}(g,S)   U(\kappa, S\setminus\{g\}, S’\setminus\{g\}) \notag\\ 
+ \sum_{g\in S\setminus S’}& \mathtt{ft}(g,S)(s_{m-|S|+1} + U(\kappa - 1, S\setminus\{g\}, S')) \label{eq : pl fpt m 1}\\
U(0, S,S’) &= 0 \qquad \forall S,S' \label{eq : pl fpt m 2}\\
U(\kappa, S,\emptyset) &= \sum_{i=1}^{\kappa} s_{m-|S|+i} \qquad\forall \kappa,S \label{eq : pl fpt m 3}
\end{align}
In Equation~\ref{eq : pl fpt m 1}, we consider all possible goods which could be placed at rank $m-|S|+1$. This occurs for good $g\in S$ with probability $\mathtt{ft}(g,S)$. 
If $g\in S'$, this good as already been picked and the agent still has to picked $\kappa$ goods within the goods in $S\setminus\{g\}$ which are ranked in last position, hence we consider $U(\kappa, S\setminus\{g\}, S’\setminus\{g\})$. 
If $g\not\in S'$, this good is picked by the agent leading to a utility $s_{m-|S|-1}$ and the agent still has to picked $\kappa-1$ goods within the goods in $S\setminus\{g\}$ which are ranked in last position, hence we consider $U(\kappa-1, S\setminus\{g\}, S’)$. Equations~\ref{eq : pl fpt m 2} and \ref{eq : pl fpt m 3} provide the base cases. 

Next, we consider the probability $P(S,S')$ that goods in $S$ are ranked in the top $|S|$ positions among a set $S'$ of goods. 
$P(S,S')$ trivially satisfies the following recursive equation. 
\begin{align*}
P(S,S') &= \sum_{g\in S} \mathtt{ft}(g,S') P(S\setminus \{s\}, S'\setminus \{s\})\\
P(\emptyset,S') &= 1
\end{align*}
Once values $U(\kappa,S,S')$ and $P(S,S')$ have been computed, we use the fact that: 
$$
\eu(\kappa,\tau) = \sum_{S'\subset G , |S'| = \tau} P(S',G) U(\kappa,G,S').
$$
To do the computation, we use memoization to store the different values $U(\kappa,S,S')$ and $P(S,S')$ (which represents $O(m4^m)$ values), avoid redundant computation, and obtain the desired time complexity.
\end{proof}
\propPLrho*
\begin{proof}
Let $\overline{m} = (m_1,m_2,\ldots,m_\rho)$ be a vector representing a set containing $m_i$ goods of value $\nu_i$. For $i \in [\rho]$, we denote by $\mathtt{ft}(i,\overline{m}) = \nu_im_i/(\sum_{j \in [\rho]} m_j\nu_{j})$, the probability that a good with parameter $\nu_i$ is ranked first among the elements of the set represented by $\overline{m}$ according to $\PL_{\nu}$. We further define $\overline{m}[-i]$ as the vector defined as $\overline{m}[-i]_i = \overline{m}_i -1$ and $\overline{m}[-i]_j = \overline{m}_j$ for $j\in [\rho]\setminus\{i\}$ and $\mathtt{sum}(\overline{m}) = \sum_{i=1}^{\rho}m_i$.

The proof relies on recursive equations. 
Let us consider the following setting.
The picker under consideration should  pick $\kappa$ goods within a set $S$ of goods occupying the $|S|$ last positions of her ranking. This set is represented by a vector $\overline{m}_S$. Goods in $G\setminus S$, occupy the top $m-S$ ranks in her ranking and have already been picked, either by her or by other agents. 
Moreover, the set $S'\subseteq S$ have already been picked by previous pickers. The set $S'$ is represented by a vector $\overline{m}_{S'}$ such that $\overline{m}_{S'} \le \overline{m}_{S}$.   
We are interested in computing the expected utility $U(\kappa,\overline{m}_{S},\overline{m}_{S'})$ of the $\kappa$ picks of the agent in such a situation. 
We argue that $U(\kappa,\overline{m}_{S},\overline{m}_{S'})$ satisfies the following recursive equations:
\begin{align}
U(\kappa, \overline{m}_{S},\overline{m}_{S'}) &= \sum_{i\in [\rho]} \ \mathtt{ft}(i,\overline{m}_{S}) (\frac{m_i'}{m_i} U(\kappa, \overline{m}_{S}[-i],\overline{m}_{S'}[-i]) \notag\\ 
+ \frac{m_i - m_i'}{m_i}&(s_{m-\mathtt{sum}(\overline{m}_S)+1} + U(\kappa-1, \overline{m}_{S}[-i],\overline{m}_{S'}))\label{eq : xp rho 1}\\
U(0, \overline{m}_{S},\overline{m}_{S'}) &= 0 \qquad \forall S,S' \label{eq : xp rho 2}\\
U(\kappa, \overline{m}_{S},\overline{m}_\emptyset) &= \sum_{i=1}^{\kappa} s_{m-\mathtt{sum}(\overline{m}_S)+i} \qquad\forall \kappa,S \label{eq : xp rho 3}
\end{align}
In Equation~\ref{eq : xp rho 1}, we consider all possible goods which could be placed at rank $m-|S|+1= m-\mathtt{sum}(\overline{m}_S)+1$, considering only their value in $\nu$. 
This occurs for a good $g$ with parameter $\nu_i$ with probability $\mathtt{ft}(i,\overline{m}_S)$. Let us assume that this good as indeed value $\nu_i$. This good is in (resp. out of) $S'$ with probability $m'_i/m_i$ (resp $(m_i-m'_i)/m_i$).
If $g\in S'$, this good as already been picked and the agent still has to picked $\kappa$ goods within the goods in $S\setminus\{g\}$ which are ranked in last position, hence we consider $U(\kappa, \overline{m}_{S}[-i],\overline{m}_{S'}[-i])$. 
If $g\not\in S'$, this good is picked by the agent leading to a utility $s_{m-|S|-1}$ and the agent still has to picked $\kappa-1$ goods within the goods in $S\setminus\{g\}$ which are ranked in last position, hence we consider $U(\kappa-1, \overline{m}_{S}[-i],\overline{m}_{S'})$. Equations~\ref{eq : xp rho 2} and \ref{eq : xp rho 3} provide the base cases. 

Next, we consider the probability $P(\overline{m},\overline{m}')$ that a set of good $S$ with vector $\overline{m}_S = \overline{m}$ are ranked in the top $\mathtt{sum}(m)$ positions among a set $S'$ of goods with vector $\overline{m}_{S'} = \overline{m}'$. 
$P(\overline{m},\overline{m}')$ trivially satisfies the following recursive equation. 
\begin{align*}
P(\overline{m},\overline{m}') &= \sum_{i\in [\rho], \overline{m}_i\neq 0} \mathtt{ft}(i,\overline{m}') P(\overline{m}[-i],\overline{m}'[-i]])\\
P(\overline{m}_\emptyset,\overline{m}') &= 1
\end{align*}
Once values $U(\kappa,\overline{m},\overline{m}')$ and $P(\overline{m},\overline{m}')$ have been computed, we use the fact that: 
$$
\eu(\kappa,\tau) = \sum_{\overline{m}' \le \overline{m}_G , \mathtt{sum}(\overline{m}') = \tau} P(\overline{m}',\overline{m}_G) U(\kappa,\overline{m}_G,\overline{m}').
$$
To do the computation, we use memoization to store the different values $U(\kappa,\overline{m}_S,\overline{m}_{S'})$ and $P(\overline{m},\overline{m}')$ (which represents $O(m\times m^{2\rho})$ values), avoid redundant computation, and obtain the desired time complexity.
\end{proof}

\section{Omitted Proofs of Section 6}\label{app : sec 6}

\begin{figure}[t]
    \centering
\begin{tikzpicture}[scale=0.65]
	\begin{axis}[
	scale = 0.8,
		xlabel = {m},
	ylabel = {\small Portion of total utility},
	const plot,
		stack plots=y,
		area style,
				ymin=0,
				axis y line = left,
        ylabel style={yshift=-0.5em},
		enlarge x limits=false]
		\foreach \n in{0,...,4}{
    \addplot table[x=m, y=U/Su\n] {./Data/5_300_Lexicographic_egalitarian_Tout.txt}
    \closedcycle;
}
	\end{axis}
\end{tikzpicture}$~~~$\begin{tikzpicture}[scale=0.65]
	\begin{axis}[
	scale = 0.8,
	const plot,
	xlabel = {m},
	ylabel = {\small Portion of goods received},
		stack plots=y,
		ymin=0,
		area style,
        axis y line = right,
        ylabel style={yshift=0.5em},
		enlarge x limits=false,
		]
		\foreach \n in{0,...,4}{
    \addplot table[x=m, y=Nb/m\n] {./Data/5_300_Lexicographic_egalitarian_Tout.txt}
    \closedcycle;
}
	\end{axis}
\end{tikzpicture}
    \caption{\small Portion of the total utility (plot on the left) and of goods (right) received by each of 5 agents with $m$ increasing from 5 to 300 in steps of 5. Maximizing ESW and using the lexicographic scoring vector and $\FI$.}
    \label{fig:results with lex}
\end{figure}
Figure~\ref{fig:results with lex} displays our results using the lexicographic scoring vector (where $s_i = 2^{m-i}$), showing the proportion of utility (left-hand side) and goods (right-hand side) obtained for $n=5$ and increasing the number of goods $m$ from 5 to 300 in steps of 5. The IC model was used and we optimized ESW. We see that because of the lexicographic scoring vector, almost all goods are given to the last picker to compensate for this disadvantageous position.

We now provide a property of optimal CSDs when maximizing ESW and for a large number of goods. 
\begin{proposition}\label{propEqualsWhenMinc}
Assume there are $n$ agents ($n$ being fixed), and let $\mathcal{K}^*_m$ be the set of allocation vectors maximizing $SW_\IC^E(\boldsymbol{k})$ when there are $m$ items. Then, if one uses the lexicographic (resp. Borda) scoring vector, then for any value $\epsilon>0$, there exists a value $M$ (dependent on $n$) such that $m \ge M$ implies that $(\max_{a\in \AgentSet} EU_{\IC}^{\boldsymbol{k}}(a) - \min_{a\in \AgentSet} EU_{\IC}^{\boldsymbol{k}}(a))/(\sum_{a\in \AgentSet} EU_{\IC}^{\boldsymbol{k}}(a)) < \epsilon$ for any (resp. an element) ${\boldsymbol{k}}\in \mathcal{K}^*_m$.   
%\cjl{Several comments about this result:
%\begin{enumerate}
%    \item it says that egalitarian social welfare not only maximizes the expected utility of the worst-off agent but, better, tends to equalize agents' utilities when $m$ is large enough. This is an interesting result (that has to be perhaps put in perspective by works, e.g. by Suksompong, that shows that it's much more likely that we obtain fair allocations when the number of goods is large wrt the number of agents).
%    \item however I would say it is more a property of egalitarian social welfare than of CSDs. Well, what we can say is that CSDs when $m$ grows tend to find the maximum utilitarian social welfare overall (provided that utilities are induced from scoring vectors), but still this is a mostly a property of USW.
%    \item in the statement of the proposition we say $\FI$, do we mean $\FI-IC$?
 %   \item because the choice of the lexicographic vector can be questioned (see py previous comment), perhaps it makes more sense to include the proof for Borda rather than with lexicographic.
 %   \item finally, why is this result in this section? 
%\end{enumerate}
%}
\end{proposition}
\begin{proof}
We treat the cases of the Borda and lexicographic scoring vectors  using two different proofs. 
\paragraph{The lexicographic case.} 
Let us fix $n$ the number of agents, a value $\epsilon>0$, and $l$ an integer such that $\epsilon/2\ge 1/2^l$. 
As $m$ ($>nl$) increases, we can ensure with a probability tending towards one that each agent receives her $l$ preferred items. 
Indeed, because of the independence and uniformity assumptions of the $\IC$ model, the probability that these sets of items do not intersect tends towards one. 
Hence, for $\epsilon>0$, there exists a value $M$ such that if $m>M$, this event (the sets being disjoint) occurs with probability $1-\epsilon/2$. 
Using the lexicographic scoring vector, this event implies that each agent will receive a proportion greater than or equal to $(1-1/2^{l})$ of the total utility she gives to items, i.e.,  $\sum_{j=1}^{m}s_m = 2^m - 1$. 
To sum up, if $m\ge M$, we can ensure that each agent receives an expected utility greater than: 
\begin{multline*}
(1-\epsilon/2)(1-1/2^{l})(2^m-1) \ge (1-\epsilon/2)^2(2^m-1)\\
\ge (1-\epsilon)(2^m-1).   
\end{multline*}

Moreover, note that this expected utility is upper bounded by $\sum_{j=1}^{m}s_m = 2^m - 1$ and that $\sum_{a\in \AgentSet} EU_{\IC}^{\boldsymbol{k}}(a)$ is lower bounded by $\sum_{j=1}^{m}s_m = 2^m - 1$. 
Hence, for any ${\boldsymbol{k}}\in \mathcal{K}^*_m$:
\begin{align*}
    \frac{|EU_{\IC}^{\boldsymbol{k}}(a_i) - EU_{\IC}^{\boldsymbol{k}}(a_j)|}{\sum_{a\in \AgentSet} EU_{\IC}^{\boldsymbol{k}}(a)} &\le  \frac{|EU_{\IC}^{\boldsymbol{k}}(a_i) - EU_{\IC}^{\boldsymbol{k}}(a_j)|}{2^m - 1}\\ 
    &\le \frac{(2^m-1) - (1-\epsilon)(2^m-1)}{2^{m}-1}\\
    &\le \epsilon.
\end{align*}

\paragraph{The Borda case.}
We wish to show that for any $\epsilon$, there always exists an optimal egalitarian solution for which the difference in portions of total utility assigned to any two different agents is smaller than $\epsilon$ when $m$ is high enough. However, instead of reasoning on the portion of total expected utility received by an agent, we will work on a proxy, denoted by  $\tilde{P}^{\boldsymbol{k}}(a) = 2EU_{\IC}^{\boldsymbol{k}}(a)/m(m+1)$. 
Note that, compared to $EU_{\IC}^{\boldsymbol{k}}(a)/(\sum_{a\in \AgentSet} EU_{\IC}^{\boldsymbol{k}}(a))$, $\tilde{P}^{\boldsymbol{k}}(a)$ replaces the expected total utility received by the $n$ agents by a lower bound on it given by $\sum_{j=1}^{m}s_m = m(m+1)/2$.

Let $\epsilon>0$ be a positive value. We set $\epsilon' = \epsilon/n$, and $m = 2/\epsilon'$. We now show by induction on $l$ the following: For any value $\tau \in \{i/m,i\in [m]_0\}$, there exists a solution ${\boldsymbol{k}}$ maximizing $\min_{a\in \{a_1,\ldots,a_l\}} \tilde{P}^{\boldsymbol{k}}(a)$ under the constraint that a proportion $\tau$ of the picks are assigned to the $l$ first pickers which ensures that  $|\tilde{P}^{\boldsymbol{k}}(a_i) - \tilde{P}^{\boldsymbol{k}}(a_j)| \le l\epsilon'$ for any $i,j \in [l]^2$. We denote by $\mathcal{P}_l^{\tau}$ the previous optimization problem and $\Gamma_l^{\tau}$ its optimal value.

The claim is trivially true for $l=1$. 
Assume, it is true for $l\ge 1$. 
We seek a solution maximizing $\min_{a\in \{a_1,\ldots,a_{l+1}\}} \tilde{P}^{\boldsymbol{k}}(a)$ given that they receive a proportion $\tau$ of the items. 
The $l$ first agents will receive a proportion $\tau' \in [0,\tau]$ of the items, and we can assume wlog that the allocation to the $l$ first agents is the one maximizing $\mathcal{P}_l^{\tau'}$ insuring our inductive property. 
Note that if $\tau'$ increases (resp. decrease) by $1/m$, this may only increase (resp. decrease) $\Gamma_l^{\tau+1/m}$  by $2/m$, i.e., $\Gamma_l^{\tau'+1/m} \le \Gamma_l^{\tau'} + 2/m$ and decrease (resp. increase) the $\tilde{P}^{\boldsymbol{k}}(a_{l+1})$ value of the $(l+1)^{th}$ picker by $2/m$ and that $\Gamma_l^{\tau}$ (resp. $\tilde{P}^{\boldsymbol{k}}(a_{l+1})$) is non-decreasing (resp. non-increasing) in $\tau$. 
Hence, by adjusting the value of $\tau'$, we can ensure that there exists a  solution ${\boldsymbol{k}}$ optimal for $\mathcal{P}_{l+1}^{\tau}$ and such that $|\tilde{P}^{\boldsymbol{k}}(a_{l+1}) - \Gamma_l^{\tau'}| \le \epsilon'$. 
Therefore, by the inductive property, $|\tilde{P}^{\boldsymbol{k}}(a_i) - \tilde{P}^{\boldsymbol{k}}(a_j)| \le (l+1)\epsilon'$ for any $i,j \in [l+1]^2$. This proves the inductive property. Using $l = n$, and $\tau = 1$, we obtain the claimed result as $$|\tilde{P}^{\boldsymbol{k}}(a_i) - \tilde{P}^{\boldsymbol{k}}(a_j)| \le \frac{|EU_{\IC}^{\boldsymbol{k}}(a_i) - EU_{\IC}^{\boldsymbol{k}}(a_j)|}{\sum_{a\in \AgentSet} EU_{\IC}^{\boldsymbol{k}}(a)}.$$ 

\end{proof}

\paragraph{A note on computation times} All tests were run with Python 3.10.12 on a personal computer with Ubuntu 22.04.4 LTS, 8 Intel(R) Core(TM) CPU i7-1185G7 3.00GHz cores and %a total of 
32 GB RAM. 
% \cgm{J'ai utilisé deux ordinateurs différents. Il y a $3$ algos à tester (?) : l'algo de prog dyn (avec et sans sampling) et l'algo avec greedy-esw à tester. Je réutilise une lib pour le sampling et ça semble être négligeable en temps. Je peux tout relancer sur le même pour être cohérant mais va-t-on garder ce paragraphe ? (pour la place)}
%As the program consists of a simple dynamic programming procedure the.
With $n=5$ and $m=70$, and a sample size of 1000 profiles, the computation of the optimal allocation using Equation~\ref{eq : DP1 utilitarian}, given that 
%Condition~\ref{condition : dp} 
\bjl prefix independence \ejl is satisfied, takes approximately 50 seconds. By contrast, the use of Algorithm~\ref{alg:esw} (GreedyESW) reduces the computation time to approximately 7 seconds. Furthermore, when applicable, the exact computation using Equation~\ref{eq : rec formula} is highly efficient, requiring only approximately 0.07 seconds.

\section{Finding a suitable scoring vector} 
%experiment}
\label{app : expe}

A question that has been overlooked until now\footnote{Not only in this paper but also in previous papers 
on 
%picking sequences, and more generally 
fair division who also use scoring vectors (e.g \cite{BaumeisterBLNNR17}).} is, where does the vector of scores come from? In order to address it we suggest, and test, the following methodology. For the specific domain at hand, prepare a questionnaire where some users, considered representative of the population of users, are presented a set of items: for instance, if the problem is about allocating time slots for using a tennis court, users are presented several time slots.
%and so on. 
Once this scoring vector has been elicited,
%once for all with a representative set of users;
it is used to determine optimal CSDs, which can be applied many times, with different sets of users. 
%For the sake of example, we chose the following domain: allocating ice-cream scoops.  
A similar method has been used for voting by \citet{BoutilierCHLPS15} (see Section 5.6). 

We designed an online experiment. To each user taking part in it,
we present 12 ice-cream flavours uniformly selected among 62 possible and elicit their utility on a scale [0,100].
% we present several ice-cream flavours and elicit their (normalized) utility on a scale [0,100]. Each user is presented the same number of flavours (12) but the set of flavours varies from one to another (they are selected by randomizing uniformly among 62 possible flavours). %

%Before values for the different flavours are elicited one after the other, 
We first ask the user to tell which is their preferred flavour ($PF$) among the 12, and we tell them that the value for $PF$ is fixed to 100. Then, for each flavour $F$ (including $PF$), we present the user a slider, with which they indicate the value of $F$ between 0 and 100. 
%They are told that they can interpret the chosen value $V$ as the exact point where they are indifferent between receiving $PF$ with probability $\frac{V}{100}$ and nothing with probability $1-\frac{V}{100}$, or receiving $F$ for sure.
%\footnote{The users who took part to the experiment are members of an academic lab (faculty members and students); most of them are familiar with decision theory.}
%which allows avoiding the complicated process of incremental elicitation.} 

Once all vectors are collected, we rearrange them non-increasingly.
%For instance if user 1 has provided the values (40, 100,  40, 100, 0, 0, 80, 70, 40, 100) then the reordered vector is (100, 100, 100, 80, 70, 40, 40, 40, 0, 0).%
Then all vectors are averaged among all users. We obtain a scoring vector $s = (s_1, \ldots, s_{12})$: $s_i$ 
%(which should in theory be equal to 100) 
is the average value, among all users, of their $i$th most preferred item. %\footnote{In theory, $s_1$ should be equal to 100; in practice it is slightly inferior, because a few users forgot the indication that their preferred item should have value 100.}

We had 54 participants. Screenshots of the experiment, information on how consents and the data were collected, as well as the list of %are included in the Appendix.
%\footnote{As the language of the questionnaire was not English, we translated it so as not to reveal any information about authorship.} 
the 54 gathered vectors are included in the Appendix. Their average is
$$\begin{array}{cccc}
  s_1 = 91.4 & s_2 = 76.6
  & s_3 = 68.2 & s_4 = 56.9\\
  s_5 =48.6 & s_6 = 41 &
  s_7 =34.3 & s_8 = 26.1\\
  s_9 = 21.1 & s_{10} = 16.5 &
  s_{11} = 10.2 & s_{12} = 5.3  
\end{array}$$
%5.196428571428571, 9.928571428571429, 15.5, 22.053571428571427, 29.785714285714285, 40.642857142857146, 50.625, 63.017857142857146, 69.55357142857143, 79.28571428571429, 85.92857142857143, 97.08928571428571)
%The average value of $s_1$ is not 100 because eight users out of 56 did not give 100 to their preferred item (possibly because they did not perfectly understand the task, or because they read too quickly).% Figure \ref{fig:exp} 
%\jerome{si on l'inclut --- je pense que ce serait bien; et dans ce cas on vire la fin de cette phrase.} 
Figure \ref{fig:exp} shows how this vector compares to the Borda vector (rescaled such that the score of one's preferred flavor is 100)\footnote{Note that by averaging we lost some interesting information about variance: some users have rapidly decreasing, and some others slowly decreasing valuations. Moreover, note that the highest score of the averaged vector is not 100 as several users did not respect this constraint.}. %\MG{C'est pas vraiment borda on le spécifie ou pas ?} vector.
%the most extreme ones are (100  36  34  21  18   8   8   0   0   0   0   0) and (100 100 100 100 100  97  97  95  95  90  90  80).
%(which, here, would be $s_1 = 100$, $s_2 = 91.7$, $s_3 = 83.3$, $s_4 = 75$, $s_5 = 66.7$, $s_6 = 58.3$, $s_7 = 50$, $s_8 = 41.7$, $s_9 = 33.3$, $s_{10} = 25$, $s_{11} = 16.7$, $s_{12} = 8.3$).

\begin{figure}
    \centering
    \begin{tikzpicture}[scale=0.6]
\begin{axis}[
title={Comparison of Borda and Experiment vector},
xlabel={s},
ylabel={Utility},
]

\addplot table[x=S, y=Exp] {./Data/Exp_Borda.txt};
\addplot table[x=S, y=Borda] {./Data/Exp_Borda.txt};
\legend{Experiment,Borda}

\end{axis}
\end{tikzpicture}
    \caption{\small Comparison of the Borda scoring score (rescaled to [0,100]) with the scoring vector obtained through our experiment by taking the expectation of the participants' answers.}
    \label{fig:exp}
\end{figure}

\subsection*{The ice-cream experiment}
We provide additional information on the experiment used for estimating an adequate scoring vector for the application of allocating ice creams with different flavors. 

The list of all ice-cream flavors used for the experiment is the following one : \\

$[$Kiwi, Litchi, Mango, Mandarin, Melon, Mirabelle, Blackberry, Blueberry, Orange, Blood orange, Apricot, Pineapple, Banana, Lemon, Lime, Cherry, Cassis, Raspberry, Coco, Fig, Strawberry, Passion fruit, Pear, Rhubarb, Grapefruit, Honey-Pine nuts, Tiramisu, Chocolate ginger, Tagada strawberry, Nougat, Speculoos, Coffee, Milk jam, Pistachio, Licorice, Lavender, Caramel, Dragibus, Avocado, Chewing gum, Olive, Chili chocolate, Tomato-Basil, Cinnamon, White chocolate, Chocolate, Almond, Poppy, Cookies, Gingerbread, Cactus, Beer, Oreo, Nutella, Vanilla, Candy floss, Rum-Raisin, Pumpkin, Chestnut, Wild pollen, Rice pudding, Salted butter caramel$]$.\\ 

Each participant was presented a subset of 12 of these flavors, sampled uniformly at random with a seed set according to the $\mathtt{Math.Random()}$ Javascript method. Note that this method sets the seed for simulating randomness in a way that depends on the browser of the user. 

We tested two different ways of explaining the experiment to participants. Indeed, we wanted to evaluate the impact of describing the experiment in one way or another.
\begin{itemize}
\item For half of the participants\footnote{In fact, each participant had a probability 0.5 of getting one set of directives or the other. Luckily this procedure split the set of participants in two sets of equal sizes.}, we presented the scores assigned to the ice-cream flavours as Von Neumann-Morgenstein utility values. We first ask the user to tell which is their preferred flavour ($PF$) among the 12, and we tell them that the value for $PF$ is fixed to 100. Then, for each flavour $F$ (including $PF$), we present the user a slider, with which they will indicate the value of $F$ between 0 and 100. They are told that they can interpret the chosen value $V$ as the exact point where they are indifferent between receiving $PF$ with probability $\frac{V}{100}$ and nothing with probability $1-\frac{V}{100}$, or receiving $F$ for sure. A screenshot of this process is reported in Figure~\ref{fig : screen_loterie}. 
\item For the other half of participants, we tested  simpler directives. As done previously, we explain to users that the value for their preferred flavour should be fixed to 100. Then, for each flavour $F$, we present the user a slider, with which they should indicate the value of $F$ between 0 and 100, without mentioning any probabilistic interpretation for these values. A screenshot of this process is reported in Figure~\ref{fig : screen_intuition}. 
\end{itemize}

The $27$ scoring vectors of the participants who followed the first (resp. second) directives are displayed on Table~\ref{tab:vectors with prob directives} (resp. Table~\ref{tab:vectors with non prob directives}).

\begin{table}[h]
\resizebox{0.8\linewidth}{!}{
\begin{tabular}{cccccccccccc}
100 & 95 & 90 & 30 & 20 & 0 &  0 &  0 &  0 &  0 &  0 &  0\\
100 & 12 & 11 &  0 &  0 &  0 &  0 &  0 &  0 &  0 &  0 &  0\\
100 & 100 &100 &100 & 93 & 90 & 73 & 59 & 52 & 48 & 10 &  6\\
100 & 93 & 93 & 90 & 63 & 50 & 33 & 13 & 10 & 10 & 0 & 0\\
100 & 100 & 100 & 70 & 50 & 50 & 40 & 40 & 30 & 30 & 30 & 25\\
91 & 75 & 59 & 17 & 17 & 15 & 12 &  5 &  0 &  0 &  0 &  0\\
99 & 90 & 86 & 80 & 61 & 60 & 32 & 30 & 24 & 21 & 18 & 17\\
100 & 95 & 90 & 80 & 60 & 50 & 10 & 10 &  5 &  5 &  0 &  0\\
100 & 95 & 90 & 80 & 70 & 60 & 60 & 50 & 45 & 30 & 25 & 10\\
100 & 95 & 85 & 80 & 80 & 70 & 60 & 52 & 40 & 30 & 20 & 10\\
100 &  80 & 75 & 60 & 50 & 15 & 10 &  5 &  5 &  3 &  1 &  0\\
100 & 40 & 30 & 20 & 20 & 10 & 10 &  5 &  5 &  5 &  2 & 0\\
100 & 90 & 79 & 76 & 76 & 66 & 64 & 51 & 32 & 26 & 18 & 10\\
100 & 100 & 90 & 90 & 90 & 89 & 88 & 86 & 84 & 80 & 75 & 59\\
100 & 70 & 60 & 60 & 40 & 30 & 20 & 15 & 10 & 10 & 10 & 10\\
95 & 93 & 92 & 54 & 50 & 42 & 37 & 27 & 25 & 23 & 20 & 0\\
100 & 30 & 30 & 30 & 20 & 20 & 20 & 10 & 10 & 0 &  0 & 0\\
100 & 98 & 94 & 94 & 94 & 92 & 92 & 92 & 90 & 80 & 68 & 10\\
100 & 43 & 36 & 35 & 31 & 27 & 18 & 17 & 14 &  8 &  5 & 0\\
100 & 85 & 80 & 70 & 25 & 15 & 15 & 10 &  5 &  5 &  2 & 0\\
100 & 70 & 70 & 50 & 50 & 30 & 30 &  0 &  0 &  0 &  0 & 0\\
100 & 82 & 81 & 76 & 75 & 74 & 53 & 31 & 21 & 14 &  0  & 0\\
100 & 95 & 79 & 73 & 56 & 50 & 48 & 45 & 33 & 15 & 10 & 9\\
 90 & 88 & 81 & 61 & 53 & 50 & 50 & 50 & 50 & 50 & 20 & 4\\
 100 & 75 & 53 & 47 & 44 & 36 & 30 & 23 & 12 &  5 &  0 & 0\\
100 & 90 & 80 & 51 & 50 & 50 & 50 & 40 & 40 & 20 & 10 &  2\\
100 & 80 & 80 & 70 & 60 & 60 & 60 & 50 & 40 & 30  &30 & 10\\
\end{tabular}}
    \caption{List of the 27 scoring vectors (one per row) obtained for participants following the directives mentioning a probabilistic interpretation for values assigned to ice-cream flavours.}
    \label{tab:vectors with prob directives}
\end{table}

\begin{table}[h]
\resizebox{0.8\linewidth}{!}{
\begin{tabular}{cccccccccccc}
83 & 82 & 80 & 78 & 73 & 70 & 65 & 65 & 57 & 50 & 13 & 12\\
91 & 87 & 76 & 55 & 42 & 41 & 40 & 37 & 36 & 18 & 15 & 10\\
100 & 80 & 70 & 60 & 60 & 50 & 50 & 40 & 40 & 30 & 20 & 10 \\
60 & 45 & 40 & 19 & 0 & 0 & 0 & 0 & 0 & 0 & 0 &  0\\
80 & 80 & 73 & 73 & 65 & 59 & 56 & 34 &  0 & 0 & 0 & 0\\
83 & 81 & 78 & 77 & 76 & 69 & 68 & 25 & 13 & 6 & 4 & 3\\
91 & 91 & 80 & 61 & 54 & 51 & 19 & 11 & 10 & 10 &  0 &  0\\
86 & 82 & 74 & 72 & 48 &  6 &  0 &  0 &  0 &  0 &  0 &  0\\
85 & 70 & 13 & 11 &  8 &  7 &  0 &  0 &  0 & 0 & 0 & 0\\
49 & 49 & 40 & 35 & 30 & 20 & 14 & 10 & 10 &  5 &  0 &  0\\
76 & 70 & 68 & 60 & 57 & 55 & 55 & 55 & 40 & 40 & 25 & 10\\
100 & 65 & 0 &  0 &  0 &  0 &  0 &  0 &  0 &  0 &  0 &  0\\
89 & 86 & 82 & 80 & 73 & 71 & 69 & 61 & 60 & 52 & 43 & 34\\
90 & 80 & 72 & 63 & 55 & 54 & 13 &  9 &  6 &  0 &  0 &  0\\
100 & 99 & 80 & 70 & 65 & 60 & 51 &  4 &  3 &  2 &  1 &  0\\
100 & 92 & 88 & 78 & 73 & 62 & 50 & 41 & 40 & 28 &  7 &  3\\
97 & 92 & 77 & 59 & 49 & 43 & 33 & 30 & 18 & 0 & 0 & 0\\
100 & 78 & 70 & 21 &  0 &  0 &  0 &  0 &  0 &  0 &  0 &  0\\
79 & 59 & 52 & 43 & 32 & 20 & 15 & 14 & 12 & 12 & 11 & 10\\
82 & 78 & 78 & 75 & 71 & 49 & 32 & 16 &  8 &  0 &  0 &  0\\
100 & 80 & 80 & 70 & 70 & 60 & 60 & 60 & 50 & 50 & 20 &  0\\
81 & 40 & 30 & 22 & 13 &  0 &  0&   0&   0&   0 &  0 &  0\\
72 & 64 & 61 & 61 & 41 & 34 & 31 & 28 & 26 & 18 & 12 & 11\\
91 & 86 & 74 & 67 & 64 & 62 & 57 & 15 & 10 &  9 &  0 & 0\\
81 & 75 & 70 & 58 & 55 & 27 & 16 & 0 &  0 &  0 &  0 &  0\\
85 & 60 & 59 & 41 & 33 & 31 & 31 & 29 & 15 &  8 &  4 &  2\\
30 & 27 & 24 & 20 & 18 & 14 & 13 & 8 &  6 &  3 &  0 &  0\\
\end{tabular}}
    \caption{List of the 27 scoring vectors (one per row) obtained for participants following the directives which did not mention a probabilistic interpretation for values assigned to ice-cream flavours.}
    \label{tab:vectors with non prob directives}
\end{table}
\begin{figure}[b]
    \centering
    \begin{tikzpicture}[scale=0.6]
\begin{axis}[
xlabel={s},
ylabel={Utility},
]

\addplot table[x=S, y=Lottery] {./Data/Exps.txt};
\addplot table[x=S, y=Simpler] {./Data/Exps.txt};
\legend{Lottery-based,Simpler}

\end{axis}
\end{tikzpicture}
    \caption{\small Scoring vectors obtained through our experiment. The scoring vector ``Lottery-based'' (resp. ``Simpler'') was obtained by averaging the answers of the participants receiving the directives which mentioned (resp. did not mention) a probabilistic interpretation for values assigned to ice-cream flavours.}
    \label{fig:exp2modes}
\end{figure}
A comparison of the resulting averaged scoring vectors with the Borda scoring vector is provided in Figure~\ref{fig:exp2modes}. We observe that the two averaged scoring vectors obtained for each set of directives have a similar shape with one being slightly above the other one. Indeed, more participants following the simpler directives did not follow the constraint that their preferred flavour among the 12 should receive a value of 100, leading to smaller values in the averaged scoring vector. This is probably due to an ambiguity about the fact that the preferred flavour of the participant should be understood as the most preferred one among the ones which are presented. While this finding points out a possible improvement for our experiment, we insist on the fact that it should be understood here as a proof of concept illustrating the feasibility of such an approach to estimate a relevant scoring vector for the domain at hand. As the results were similar for the two ways of describing the experiment, we decided to merge the two list of vectors for plotting the Figure~\ref{fig:exp} presented in the main document of the submission.

\begin{figure*}[h!]
    \centering
    \resizebox{\linewidth}{!}{\includegraphics{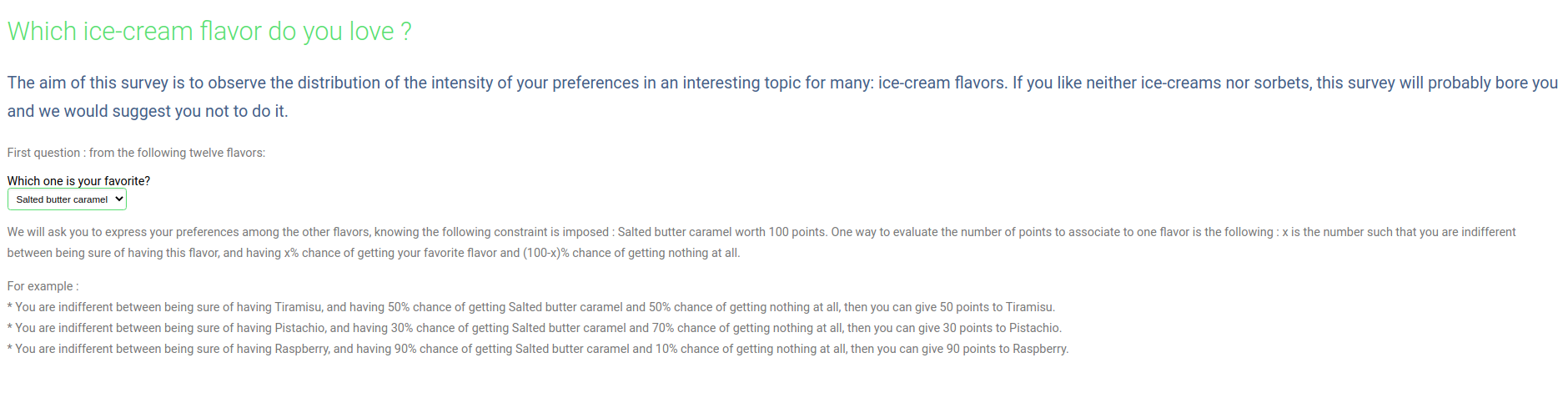}}
    \caption{Screenshot of our questionnaire. Directives to the user mentioning a probabilistic interpretation for the scores assigned to ice-cream flavors.}
    \label{fig : screen_loterie}
\end{figure*}

\begin{figure*}[h]
    \centering
    \resizebox{\linewidth}{!}{\includegraphics{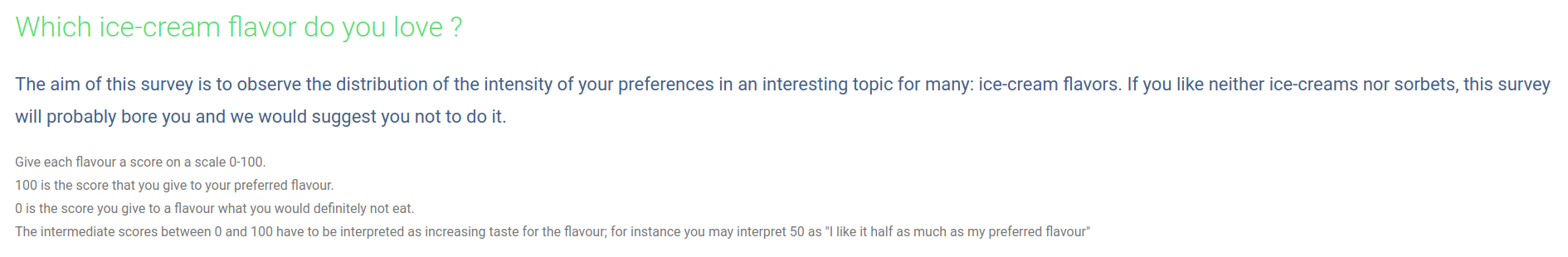}}
    \caption{Screenshot of our questionnaire. Directives to the user not mentioning a probabilistic interpretation for the scores assigned to ice-cream flavors.}
    
    \label{fig : screen_intuition}
\end{figure*}

\begin{figure*}[h!]
    \centering
    \resizebox{\linewidth}{!}{\includegraphics{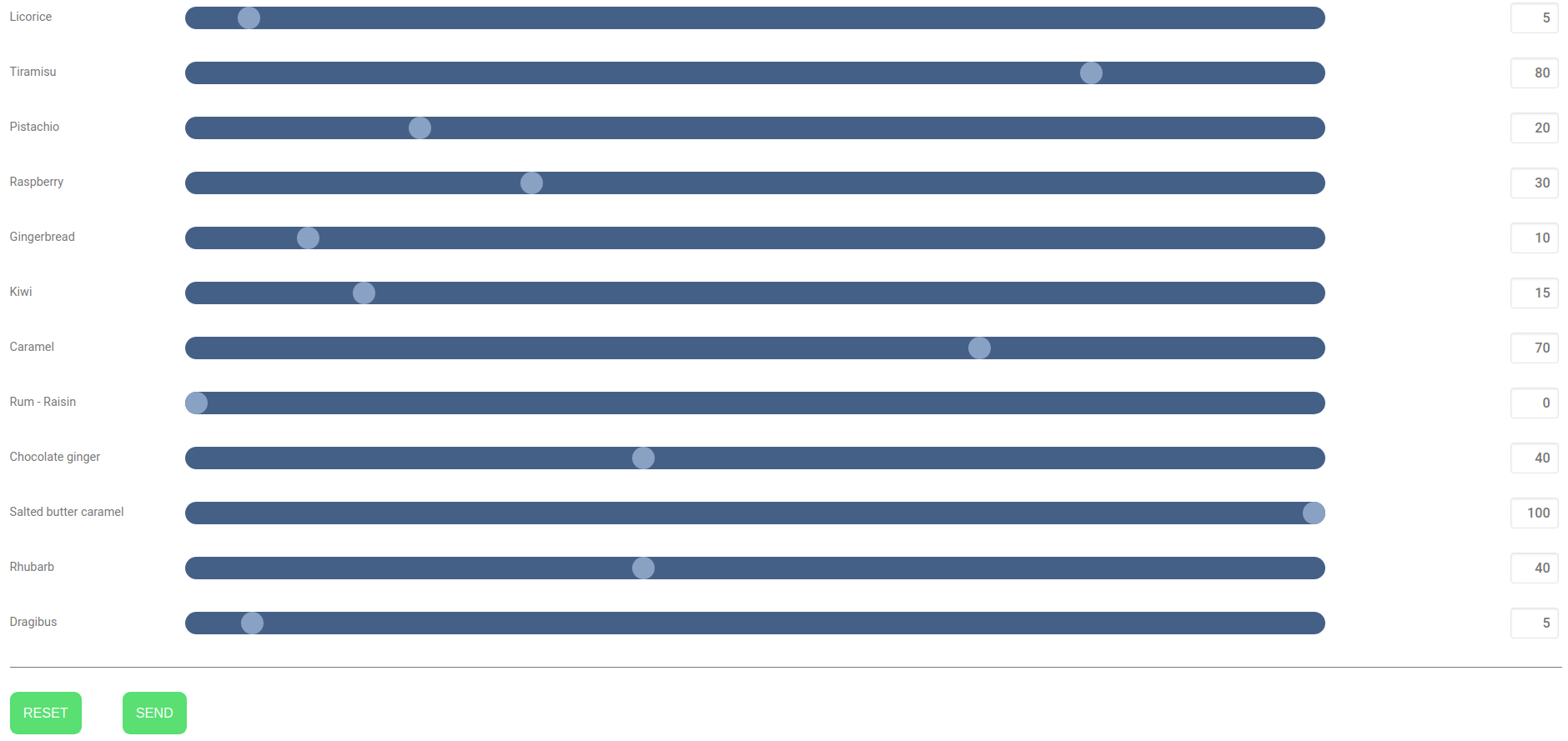}}
    \caption{Screenshot of our questionnaire. The sliders make it possible for the participants to assign a value to each flavor.}
    \label{fig : screen_slide}
\end{figure*}

\paragraph{Collect of Consents and Data} 
The data which was collected was completely anonymous; indeed, we did not collect any piece of information about the participants besides the scores assigned to the ice-cream flavors. Moreover, we checked with the ethic committee of one of the authors' university that the experiment complies with the data protection regulation. Last, as illustrated in Figure~\ref{fig : consent}, all the participants to the experiment had to check a box, confirming that they agreed that their answers would be stored and used for research purposes.

\begin{figure*}[h!]
    \centering
    \resizebox{\linewidth}{!}{\includegraphics{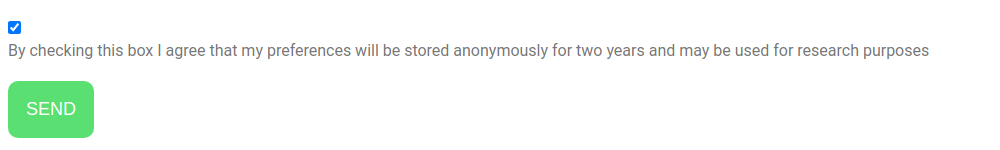}}
    \caption{Screenshot of our questionnaire. Participants had to check a box, agreeing that their answers could be used for research purposes.}
    
    \label{fig : consent}
\end{figure*}

%Another useful quantity that we derived from these data is where the dataset is positioned on the axis from full correlation to full independence. For this we did the following: for each pair of flavours $\{x,y\}$, we consider all users who have both in their selection and who do not give the same value to $x$ and $y$. Let  $w(x,y)$ as the average proportion of these individuals preferring $x$ to $y$, and $w(y,x) = 1 - w(x,y)$. Let $\alpha(x,y) = |w(x,y)-w(y,x)| \in [0, 1]$, measuring the degree in which $x$ and $y$ are far from each other in the users' preferences. Finally, we compute the concentration parameter $\alpha$ as the average value of $\alpha(x,y)$ for all $\{x,y\}$. We found $\alpha = 0.74$.
%sNote that full correlation corresponds to $\alpha = 1$ and full independence to $\alpha = 0$, therefore for this specific domain we are closer to full correlation than to full independence. Finally, since a mixture $\FIC{\lambda}$ is associated with concentration parameter  $1-\lambda$, for our specific domain it makes sense to choose the mixed distribution $\FIC{0.26}$.

\section{The price of the assignment of agents to positions}\label{app : AAP}

By abuse of notation, we may also use notation $SW_{\boldsymbol{P}}^x(\boldsymbol{k})$ (for $x\in\{U,E,N\}$) when $\ProbModel$ is the degenerate probability distribution for which profile $\boldsymbol{P}$ occurs with probability 1. 

So far, we have considered probability distributions over profiles that treat all agents in an interchangeable way. Hence, deciding who should be agent $a_1$ and pick first, who should be agent $a_2$ and pick second {\em et caetera}, has no impact on ex ante social welfare. What about its
%However, this assignment from agents to positions may have a large 
impact on ex post social welfare? 
We now study this impact, in terms of loss of social welfare between the best possible assignment and the worst possible assignments of agents to positions.
%For ease of exposition, we will study this question by focusing on utilitarian social welfare, although our results easily extend to egalitarian and Nash social welfare notions. 

Let $S_n$ denotes the set of permutations of $[n]$. Given $\pi\in S_n$ and $\boldsymbol{P}$ a preference profile, $\boldsymbol{P}_{\pi}$ denotes the preference profile obtained from $\boldsymbol{P}$ by permuting agent’s rankings according to $\pi$.

\begin{definition}
Given a CSD with vector $\boldsymbol{k}$ and a preference profile $\boldsymbol{P}$, the utilitarian, egalitarian and Nash price of assignment of agents to positions, %$\mathcal{P}_{AaP}^U$ and  $\mathcal{P}_{AaP}^e$ respectively, are
denoted by $\mathcal{P}_{AtoP}^{u}$, $\mathcal{P}_{AtoP}^{E}$, and $\mathcal{P}_{AtoP}^{N}$ respectively, are defined by: 
\begin{align*}
&\mathcal{P}_{AtoP}^U = \frac{\max_{\pi\in S_n} SW_{\boldsymbol{P}_\pi}^U(\boldsymbol{k})}{ \min_{\pi\in S_n} SW_{\boldsymbol{P}_\pi}^U(\boldsymbol{k})}, \\ &\mathcal{P}_{AtoP}^E = \frac{\max_{\pi\in S_n} SW_{\boldsymbol{P}_\pi}^E(\boldsymbol{k})}{ \min_{\pi\in S_n} SW_{\boldsymbol{P}_\pi}^E(\boldsymbol{k})},\\
&\mathcal{P}_{AtoP}^N = \frac{\max_{\pi\in S_n} SW_{\boldsymbol{P}_\pi}^N(\boldsymbol{k})}{ \min_{\pi\in S_n} SW_{\boldsymbol{P}_\pi}^N(\boldsymbol{k})}.
\end{align*}
\end{definition}

We make the two following easy observations which hold whatever the notion of social welfare which is used: 
\begin{enumerate}
    \item The worst social welfare that can be obtained when allocating resources using a CSD is obtained when for all $j\in [m]$, the $j^{th}$ good that is picked by an agent is her $j^{th}$ preferred good. In particular, this results in a utilitarian social welfare of $\sum_{j=1}^{m} s_j$.
    %where $\star = \sum, \prod$ or $\min$ depending on the chosen notion of social welfare.
    \item The best social welfare that can be obtained when allocating resources using a picking sequence (not necessarily non-interleaving) where $a_i$ picks $k_i$ goods for all $i\in [n]$ is obtained when each agent picks her $k_i$ preferred goods. This results in a social welfare of $\star_{i=1}^n\sum_{j=1}^{k_i} s_j$, where $\star = \sum, \prod$ or $\min$ depending on the chosen notion of social welfare.
\end{enumerate}
 
Hence, upper bounds on $\mathcal{P}_{AtoP}^U$, $\mathcal{P}_{AtoP}^E$, and $\mathcal{P}_{AtoP}^N$ are given by:
$$\frac{\sum_{i=1}^n\sum_{j=1}^{k_i} s_j}{\sum_{j=1}^m s_j}, 
\frac{\sum_{j=1}^{k_{\min}} s_j}{\min_{i \in [n]} \sum_{j=c_i+1}^{c_i+k_i} s_j}, 
\text{ and }\frac{\prod_{i=1}^n(\sum_{j=1}^{k_i} s_j)}{\prod_{i=1}^n(\sum_{j=c_i+1}^{c_i+k_i} s_j)}.$$
where $c_i = \sum_{l < i} k_l$ and $k_{min} = \min\{k_i | i\in [n]\}$. 

We show that there exists a preference profile and a CSD, such that this bound is closely matched.

\begin{proposition}
Assume $m = d\times n$ with $d\in \mathbb{N^*}$. There exists a preference profile and a CSD such that: 
\begin{multline*} \\
\frac{\displaystyle(n-1)\sum_{j=1}^{d} s_j + \sum_{j=d+1}^{2d} s_j}{\displaystyle\sum_{j=1}^m s_j} \le \mathcal{P}_{AtoP}^U \le \frac{\displaystyle n\sum_{j=1}^{d} s_j}{\displaystyle\sum_{j=1}^m s_j},\\
\frac{\displaystyle\sum_{j=d+1}^{2d} s_j}{\displaystyle\sum_{j=(n-1)d+1}^m s_j} \le \mathcal{P}_{AtoP}^E \le \frac{\displaystyle\sum_{j=1}^{d} s_j}{\displaystyle\sum_{j=(n-1)d+1}^m s_j},\\
\frac{\displaystyle\left({\sum_{j=1}^{d} s_j}\right)^{(n-1)} \times \sum_{j=d+1}^{2d} s_j}{\displaystyle\prod_{i=1}^n\sum_{j=(i-1)d+1}^{id} s_j} \le \mathcal{P}_{AtoP}^N \le \frac{\displaystyle\left({\sum_{j=1}^{d} s_j}\right)^n}{\displaystyle\prod_{i=1}^n\sum_{j=(i-1)d+1}^{id} s_j}.\\
\end{multline*}
\end{proposition}
\begin{proof}
%Let us assume $m = d\times n$. 
Consider the CSD with vector $\boldsymbol{k}$ such that $k_1 = k_2 = \ldots = k_n = d$. %$(\underbrace{\lceil m/n \rceil, \ldots, \lceil m/n \rceil}_{r\times},\underbrace{\lfloor m/n\rfloor, \ldots,\lfloor m/n\rfloor}_{(n-r)\times})$.
Given the previously defined vector $\boldsymbol{k}$, we build a preference profile $\boldsymbol{P}$ as follows. 
Let $S_i$ be a set of $d$ goods for $i \in \{2, \ldots, n\}$ such that $S_i\cap S_j = \emptyset$ for all $i \neq j \in \{2, \ldots, n\}$, $S_1 = S_2$, and $S_{n+1} = \ItemSet\setminus \bigcup_{i=2}^n S_i$. 
We let $S_i$ be the preferred goods of agent $a_i$ for $i\in[n]$. Additionally, each agent $a_j$ with $j\in [n]$ prefers any good in $S_s$ to any good in $S_t$ if $s < t$ and $s,t\in [n]\setminus\{j\}$. 
Lastly, we assume each agent $a_j$ with $j\in \{2,\ldots,n\}$ ranks goods in  $S_{n+1}$ last and that $a_1$ ranks these goods just after the ones in $S_1$.

One can easily check that for all $i\in [n]$, $U_{\boldsymbol{P}}^{\boldsymbol{k}}(a_i) = \sum_{j=(i-1)d+1}^{id} s_j$. If we otherwise consider the permutation $\pi = (2,3,\ldots,n,1)$, we obtain that the first $n-1$  agents get utility $\sum_{j=1}^{d} s_j$ while the last agent of the sequence (agent $a_1$) gets utility $\sum_{j=d+1}^{2d} s_j$. The three lower bounds follow.  
\end{proof}

To give an example, let us take the Borda scoring vector. The price of assignment of agents to positions is reasonable for utilitarianism, as it tends to 2 when $m$ grows; it is much larger for egalitarianism (it is in the order of $m$ when $m$ grows), with Nash being even worse (especially if $n$ grows and $d$ is kept constant, $\mathcal{P}_{AtoP}^N$ explodes). As a consequence, optimizing the CSD gives good {\em ex ante fairness} guarantees, but much less {\em ex post fairness} guarantees (it is consistent with the well-known general observation, in fair division, that ex post fairness guarantees are harder to obtain than ex post fairness guarantees).

\section{Examples}\label{app : examples}

All examples were computed using a sample of $k=1000$ preference profiles drawn from the distribution of the table. Note that with $\FC$ and utilitarianism, every allocation is optimal.

\begin{table}[H]
\caption{Results for $\FC$}
\centering
\begin{tabular}{|c|c|c|c|}
\hline
 (n,m) & SW & Best Policy & Best Utilities \\ \hline
 & ESW & (1, 3) & [4.0, 6.0] \\ \cline{2-4}
(2,4) & NSW & (1, 3) & [4.0, 6.0] \\ \cline{2-4}
 & USW & (4, 0) & [10.0, 0.0] \\ \hline
 & ESW & (2, 5) & [13.0, 15.0] \\ \cline{2-4}
(2,7) & NSW & (2, 5) & [13.0, 15.0] \\ \cline{2-4}
 & USW & (7, 0) & [28.0, 0.0] \\ \hline
 & ESW & (3, 7) & [27.0, 28.0] \\ \cline{2-4}
(2,10) & NSW & (3, 7) & [27.0, 28.0] \\ \cline{2-4}
 & USW & (10, 0) & [55.0, 0.0] \\ \hline
\hline
 & ESW & (1, 1, 2) & [4.0, 3.0, 3.0] \\ \cline{2-4}
(3,4) & NSW & (1, 1, 2) & [4.0, 3.0, 3.0] \\ \cline{2-4}
 & USW & (4, 0, 0) & [10.0, 0.0, 0.0] \\ \hline
 & ESW & (1, 2, 4) & [7.0, 11.0, 10.0] \\ \cline{2-4}
(3,7) & NSW & (1, 2, 4) & [7.0, 11.0, 10.0] \\ \cline{2-4}
 & USW & (7, 0, 0) & [28.0, 0.0, 0.0] \\ \hline
 & ESW & (2, 3, 5) & [19.0, 21.0, 15.0] \\ \cline{2-4}
(3,10) & NSW & (2, 3, 5) & [19.0, 21.0, 15.0] \\ \cline{2-4}
 & USW & (10, 0, 0) & [55.0, 0.0, 0.0] \\ \hline
\hline
 & ESW & (1, 1, 1, 1) & [4.0, 3.0, 2.0, 1.0] \\ \cline{2-4}
(4,4) & NSW & (1, 1, 1, 1) & [4.0, 3.0, 2.0, 1.0] \\ \cline{2-4}
 & USW & (4, 0, 0, 0) & [10.0, 0.0, 0.0, 0.0] \\ \hline
 & ESW & (1, 1, 2, 3) & [7.0, 6.0, 9.0, 6.0] \\ \cline{2-4}
(4,7) & NSW & (1, 1, 2, 3) & [7.0, 6.0, 9.0, 6.0] \\ \cline{2-4}
 & USW & (7, 0, 0, 0) & [28.0, 0.0, 0.0, 0.0] \\ \hline
 & ESW & (2, 2, 2, 4) & [19.0, 15.0, 11.0, 10.0] \\ \cline{2-4}
(4,10) & NSW & (1, 2, 2, 5) & [10.0, 17.0, 13.0, 15.0] \\ \cline{2-4}
 & USW & (10, 0, 0, 0) & [55.0, 0.0, 0.0, 0.0] \\ \hline
\hline
\end{tabular}
\end{table}

\newpage

\begin{table}[H]
\caption{Results for $\IC$}
\centering
\begin{tabular}{|c|c|c|c|}
\hline
 (n,m) & SW & Best Policy & Best Utilities \\ \hline
 & ESW & (2, 2) & [7.0, 4.97] \\ \cline{2-4}
(2,4) & NSW & (2, 2) & [7.0, 4.97] \\ \cline{2-4}
 & USW & (2, 2) & [7.0, 4.97] \\ \hline
 & ESW & (3, 4) & [18.0, 16.02] \\ \cline{2-4}
(2,7) & NSW & (3, 4) & [18.0, 16.02] \\ \cline{2-4}
 & USW & (4, 3) & [22.0, 12.04] \\ \hline
 & ESW & (4, 6) & [34.0, 33.05] \\ \cline{2-4}
(2,10) & NSW & (4, 6) & [34.0, 33.05] \\ \cline{2-4}
 & USW & (5, 5) & [40.0, 27.59] \\ \hline
\hline
 & ESW & (1, 1, 2) & [4.0, 3.75, 4.97] \\ \cline{2-4}
(3,4) & NSW & (1, 1, 2) & [4.0, 3.75, 4.97] \\ \cline{2-4}
 & USW & (2, 1, 1) & [7.0, 3.34, 2.45] \\ \hline
 & ESW & (2, 2, 3) & [13.0, 12.0, 11.85] \\ \cline{2-4}
(3,7) & NSW & (2, 2, 3) & [13.0, 12.0, 11.85] \\ \cline{2-4}
 & USW & (3, 2, 2) & [18.0, 11.19, 7.96] \\ \hline
 & ESW & (3, 3, 4) & [27.0, 24.86, 22.11] \\ \cline{2-4}
(3,10) & NSW & (3, 3, 4) & [27.0, 24.86, 22.11] \\ \cline{2-4}
 & USW & (4, 3, 3) & [34.0, 23.76, 16.64] \\ \hline
\hline
 & ESW & (1, 1, 1, 1) & [4.0, 3.74, 3.35, 2.46] \\ \cline{2-4}
(4,4) & NSW & (1, 1, 1, 1) & [4.0, 3.74, 3.35, 2.46] \\ \cline{2-4}
 & USW & (1, 1, 1, 1) & [4.0, 3.74, 3.35, 2.46] \\ \hline
 & ESW & (1, 2, 2, 2) & [7.0, 12.58, 11.16, 7.78] \\ \cline{2-4}
(4,7) & NSW & (1, 2, 2, 2) & [7.0, 12.58, 11.16, 7.78] \\ \cline{2-4}
 & USW & (2, 2, 2, 1) & [13.0, 11.99, 9.97, 3.9] \\ \hline
 & ESW & (2, 2, 2, 4) & [19.0, 18.34, 17.35, 22.15] \\ \cline{2-4}
(4,10) & NSW & (2, 2, 3, 3) & [19.0, 18.34, 23.59, 16.58] \\ \cline{2-4}
 & USW & (3, 3, 2, 2) & [27.0, 24.79, 15.4, 10.92] \\ \hline
\hline
\end{tabular}
\end{table}

\newpage

\begin{table}[H]
\caption{Results for $\PL_{\boldsymbol{\nu}}$ with $\boldsymbol{\nu} = (1.1^m,1.1^{m-1},\ldots,1.1^1)$}
\centering
\begin{tabular}{|c|c|c|c|}
\hline
 (n,m) & SW & Best Policy & Best Utilities \\ \hline
 & ESW & (2, 2) & [7.0, 4.96] \\ \cline{2-4}
(2,4) & NSW & (2, 2) & [7.0, 4.96] \\ \cline{2-4}
 & USW & (2, 2) & [7.0, 4.96] \\ \hline
 & ESW & (3, 4) & [18.0, 15.9] \\ \cline{2-4}
(2,7) & NSW & (3, 4) & [18.0, 15.9] \\ \cline{2-4}
 & USW & (4, 3) & [22.0, 11.92] \\ \hline
 & ESW & (4, 6) & [34.0, 32.37] \\ \cline{2-4}
(2,10) & NSW & (4, 6) & [34.0, 32.37] \\ \cline{2-4}
 & USW & (5, 5) & [40.0, 26.74] \\ \hline
\hline
 & ESW & (1, 1, 2) & [4.0, 3.74, 5.06] \\ \cline{2-4}
(3,4) & NSW & (1, 1, 2) & [4.0, 3.74, 5.06] \\ \cline{2-4}
 & USW & (2, 1, 1) & [7.0, 3.33, 2.48] \\ \hline
 & ESW & (2, 2, 3) & [13.0, 11.94, 11.73] \\ \cline{2-4}
(3,7) & NSW & (2, 2, 3) & [13.0, 11.94, 11.73] \\ \cline{2-4}
 & USW & (3, 2, 2) & [18.0, 11.08, 7.9] \\ \hline
 & ESW & (3, 3, 4) & [27.0, 24.66, 21.45] \\ \cline{2-4}
(3,10) & NSW & (3, 3, 4) & [27.0, 24.66, 21.45] \\ \cline{2-4}
 & USW & (4, 3, 3) & [34.0, 23.31, 16.0] \\ \hline
\hline
 & ESW & (1, 1, 1, 1) & [4.0, 3.76, 3.33, 2.47] \\ \cline{2-4}
(4,4) & NSW & (1, 1, 1, 1) & [4.0, 3.76, 3.33, 2.47] \\ \cline{2-4}
 & USW & (1, 1, 1, 1) & [4.0, 3.76, 3.33, 2.47] \\ \hline
 & ESW & (1, 2, 2, 2) & [7.0, 12.55, 11.17, 8.03] \\ \cline{2-4}
(4,7) & NSW & (1, 2, 2, 2) & [7.0, 12.55, 11.17, 8.03] \\ \cline{2-4}
 & USW & (2, 2, 1, 2) & [13.0, 11.98, 5.98, 8.09] \\ \hline
 & ESW & (2, 2, 2, 4) & [19.0, 18.26, 17.09, 21.33] \\ \cline{2-4}
(4,10) & NSW & (2, 2, 2, 4) & [19.0, 18.26, 17.09, 21.33] \\ \cline{2-4}
 & USW & (3, 3, 2, 2) & [27.0, 24.5, 14.89, 10.35] \\ \hline
\hline
\end{tabular}
\end{table}

\newpage

\begin{table}[H]
\caption{Results for $\Ml_{\phi,\mu}$ with $\phi=0.8$}
\centering
\begin{tabular}{|c|c|c|c|}
\hline
 (n,m) & SW & Best Policy & Best Utilities \\ \hline
 & ESW & (2, 2) & [7.0, 5.03] \\ \cline{2-4}
(2,4) & NSW & (2, 2) & [7.0, 5.03] \\ \cline{2-4}
 & USW & (2, 2) & [7.0, 5.03] \\ \hline
 & ESW & (3, 4) & [18.0, 15.45] \\ \cline{2-4}
(2,7) & NSW & (3, 4) & [18.0, 15.45] \\ \cline{2-4}
 & USW & (4, 3) & [22.0, 11.51] \\ \hline
 & ESW & (4, 6) & [34.0, 31.33] \\ \cline{2-4}
(2,10) & NSW & (4, 6) & [34.0, 31.33] \\ \cline{2-4}
 & USW & (5, 5) & [40.0, 25.79] \\ \hline
\hline
 & ESW & (1, 1, 2) & [4.0, 3.75, 4.98] \\ \cline{2-4}
(3,4) & NSW & (1, 1, 2) & [4.0, 3.75, 4.98] \\ \cline{2-4}
 & USW & (2, 1, 1) & [7.0, 3.32, 2.43] \\ \hline
 & ESW & (2, 2, 3) & [13.0, 11.86, 11.29] \\ \cline{2-4}
(3,7) & NSW & (2, 2, 3) & [13.0, 11.86, 11.29] \\ \cline{2-4}
 & USW & (3, 2, 2) & [18.0, 10.95, 7.36] \\ \hline
 & ESW & (3, 3, 4) & [27.0, 24.06, 20.22] \\ \cline{2-4}
(3,10) & NSW & (3, 3, 4) & [27.0, 24.06, 20.22] \\ \cline{2-4}
 & USW & (4, 3, 3) & [34.0, 22.61, 14.93] \\ \hline
\hline
 & ESW & (1, 1, 1, 1) & [4.0, 3.73, 3.32, 2.42] \\ \cline{2-4}
(4,4) & NSW & (1, 1, 1, 1) & [4.0, 3.73, 3.32, 2.42] \\ \cline{2-4}
 & USW & (1, 1, 1, 1) & [4.0, 3.73, 3.32, 2.42] \\ \hline
 & ESW & (1, 2, 2, 2) & [7.0, 12.51, 10.84, 7.68] \\ \cline{2-4}
(4,7) & NSW & (1, 2, 2, 2) & [7.0, 12.51, 10.84, 7.68] \\ \cline{2-4}
 & USW & (2, 2, 1, 2) & [13.0, 11.84, 5.77, 7.58] \\ \hline
 & ESW & (2, 2, 2, 4) & [19.0, 18.11, 16.69, 19.98] \\ \cline{2-4}
(4,10) & NSW & (2, 2, 2, 4) & [19.0, 18.11, 16.69, 19.98] \\ \cline{2-4}
 & USW & (3, 3, 2, 2) & [27.0, 24.02, 14.41, 9.49] \\ \hline
\hline
\end{tabular}
\end{table}

\end{document}